\documentclass[journal, letterpaper, draftcls, onecolumn, 
]{IEEEtran}

\usepackage[cmex10]{amsmath}
\usepackage{amssymb, amsfonts, amsthm}

\usepackage{enumerate}
\usepackage{graphicx}

\newcommand{\CM}{\textnormal{CM}}
\newcommand{\CONF}{\textnormal{CONF}}
\newcommand{\AUX}{\textnormal{AUX}}

\newcommand{\eps}{\varepsilon}

\newcommand{\PP}{\mathbb{P}}
\newcommand{\EE}{\mathbb{E}}

\newcommand{\K}{\mathcal{K}}
\renewcommand{\P}{\mathcal{P}}
\newcommand{\U}{\mathcal{U}}

\newcommand{\W}{\mathcal{W}}
\newcommand{\X}{\mathcal{X}}
\newcommand{\Y}{\mathcal{Y}}
\newcommand{\Z}{\mathcal{Z}}

\newcommand{\uu}{\mathbf{u}}
\newcommand{\xx}{\mathbf{x}}
\newcommand{\yy}{\mathbf{y}}
\newcommand{\zz}{\mathbf{z}}

\newcommand{\x}[2]{\xx_{#1}^{#2}}
\newcommand{\y}[2]{\yy_{#1}^{#2}}
\newcommand{\tx}[2]{\tilde\xx_{#1}^{#2}}
\newcommand{\ty}[2]{\tilde\yy_{#1}^{#2}}

\newcommand{\rund}[1]{\left(#1\right)}
\newcommand{\eckig}[1]{\left[#1\right]}
\newcommand{\betr}[1]{\left\lvert#1\right\rvert}
\newcommand{\menge}[1]{\left\{#1\right\}}

\newtheorem{theorem}{Theorem}
\newtheorem{lem}{Lemma}
\newtheorem{cor}{Corollary}

\theoremstyle{definition}
\newtheorem{ex}{Example}
\newtheorem{defn}{Definition}

\theoremstyle{remark}
\newtheorem{rem}{Remark}

\begin{document}

\title{The Compound Multiple Access Channel with Partially Cooperating Encoders}

\author{Moritz~Wiese,~\IEEEmembership{Student~Member,~IEEE,}
        Holger~Boche,~\IEEEmembership{Fellow,~IEEE,}
        Igor Bjelakovi\'{c}, 
        and~Volker~Jungnickel,~\IEEEmembership{Member,~IEEE}%
\thanks{This work was supported by the Deutsche Forschungsgemeinschaft (DFG) projects Bo 1734/15-1 and Bo 1734/16-1. The material in this paper was presented in part at the 11th IEEE International Workshop on
Signal Processing Advances in Wireless Communications (SPAWC 2010), Marrakech, Morocco, June 2010, and at the 2010 International Symposium on Information Theory and its Applications (ISITA 2010), Taichung, Taiwan, October 2010.}
\thanks{M. Wiese, H. Boche, and I. Bjelakovi\'{c} were with the Heinrich-Hertz-Lehrstuhl f\"{u}r Informationstheorie und Theoretische Informationstechnik, Technische Universit\"{a}t Berlin, Berlin, Germany. They are now with the Lehrstuhl f\"ur Theoretische Informationstechnik, Technische Universit\"at M\"unchen, Munich, Germany (e-mail: \{wiese, boche, igor.bjelakovic\}@tum.de)}
\thanks{H. Boche also was and V. Jungnickel is with the Fraunhofer Heinrich-Hertz-Institut, Berlin, Germany (e-mail: jungnickel@hhi.fraunhofer.de)}}%

\markboth{submitted to Transactions on Information Theory}%
{The Compound Multiple Access Channel with Partially Cooperating Encoders}

\maketitle

\begin{abstract}
  The goal of this paper is to provide a rigorous information-theoretic analysis of subnetworks of interference networks. We prove two coding theorems for the compound multiple-access channel with an arbitrary number of channel states. The channel state information at the transmitters is such that each transmitter has a finite partition of the set of states and knows which element of the partition the actual state belongs to. The receiver may have arbitrary channel state information. The first coding theorem is for the case that both transmitters have a common message and that each has an additional common message. The second coding theorem is for the case where rate-constrained, but noiseless transmitter cooperation is possible. This cooperation may be used to exchange information about channel state information as well as the messages to be transmitted. The cooperation protocol used here generalizes Willems' conferencing. We show how this models base station cooperation in modern wireless cellular networks used for interference coordination and capacity enhancement. In particular, the coding theorem for the cooperative case shows how much cooperation is necessary in order to achieve maximal capacity in the network considered.
\end{abstract}


\section{Introduction}\label{sect:intro}

\subsection{Motivation}\label{subsect:motiv}

In modern cellular systems, interference is one of the main factors which limit the communication capacity. In order to further enhance performance, methods to better control interference have recently been investigated intensively. One of the principal techniques to achieve this is cooperation among neighboring base stations. This will be part of the forthcoming LTE-Advanced cellular standard. It is seen as a means of achieving the desired spectral efficiency of mobile networks. In addition, it may enhance the performance of cell-edge users, a very important performance metric of future wireless cellular systems. Finally, fairness issues are expected to be resolved more easily with base station cooperation.

In standardization oriented literature, the assumptions generally are very strict. The cooperation backbones, i.e.\ the wires linking the base stations, are assumed to have infinite capacity. Full channel state information (CSI) is assumed to be present at all cooperating base stations. Then, multiple-input-multiple-output (MIMO) optimization techniques can be used for designing the system \cite{KFV}. However, while providing a useful theoretical benchmark, the results thus obtained are not accepted by the operators as reliably predicting the performance of actual networks.

In order to obtain a more realistic assessment of the performance of cellular networks with base station cooperation,  the above assumptions need to be adapted to reality. First, it is well-known that one cannot really assume perfect CSI in mobile communication networks. Second, glass fibers or any medium used for the backbones never have infinite capacity. The assumption of finite cooperation capacity will also lead to a better understanding of the amount of cooperation necessary to achieve a certain performance. Vice versa, we would like to know which capacity can be achieved with the backhaul found in heterogeneous networks using microwave, optical fibers and other media. Such insights would get lost when assuming infinite cooperation capacity. 

The question arises how much cooperation is needed in order to achieve the same performance as would be achievable with infinite cooperation capacity. For general interference networks with multiple receivers, the analysis is very difficult. Thus it is natural to start by taking a closer look at component networks which together form a complete interference network. Such components are those subnetworks formed by the complete set of base stations, but with only one receiving mobile. Then there is no more interference, so one can concentrate on finding out by how much the capacity increases by limited base station cooperation. This result can be seen as a first step towards a complete rigorous analysis of general interference networks.

A situation which is closely related can be phrased in the cooperation setting as well. Usually, there is only one data stream intended for one receiver. Assume that a central node splits this data stream into two components. Each of these components is then forwarded to one of two base stations. Using the cooperation setting, one can address the question how much overhead needs to be transmitted by the splitter with the data component, i.e. how much information about the data component and the CSI intended for one base station needs to be known at the other base station in order to achieve a high, possibly maximal data rate.

In \cite{MJH}, the cooperation of base stations in an uplink network is analyzed. A turbo-like decoding scheme is proposed. Different degrees of cooperation and different cooperation topologies are compared in numerical simulations. In \cite{JETAL}, work has also been done on the practical level to analyze cooperative schemes. The implementation of a real-time distributed cooperative system for the downlink of the fourth-generation standard LTE-Advanced was presented. In that system, the channel state information (CSI) at the transmitters was imperfect, the limited-capacity glass fibers between the transmitting base stations were used to exchange CSI and data information. A feeder distributed the data among the transmitting base stations.

A question which is not addressed in this work but which will be considered in the future is what rates can be achieved if there are two networks as described above which belong to different providers and which hence do not jointly optimize their coding, to say nothing of active cooperation. In that case, uncontrolled interference heavily disturbs each network, and challenges different from those considered here need to be faced by the system designer.

\subsection{Theory}

The rigorous analysis of such cellular wireless systems as described above using information-theoretic methods should provide useful insights. The ultimate performance limits as well as the optimal cooperation protocols can be derived from such an analysis. The first information-theoretic approach to schemes with cooperating encoders goes back to Willems \cite{Wi1, Wi2} long before this issue was relevant for practical networks. For that reason, it was not considered much in the next two decades. Willems considers a protocol where before transmission, the encoders of a discrete memoryless Multiple Access Channel (MAC) may exchange information about their messages via noiseless finite-capacity links (one in each direction). This may be done in a causal and iterative fashion, so the protocol is called a conferencing protocol. 

For the reasons mentioned at the beginning, Willems' conferencing protocol has attracted interest in recent years. Gaussian MACs using Willems conferencing between the encoders were analyzed in \cite{BLW} and \cite{Wig}. Moreover, in these two works, it was shown that interference which is known non-causally at the encoders does not reduce capacity. For a compound MAC, both discrete and Gaussian, with two possible channel realizations and full CSI at the receiver, the capacity region was found in \cite{MYK}. In the same paper, the capacity region was found for the interference channel if only one transmitter can send information to the other (unidirectional cooperation) and if the channel is in the strong interference regime. Another variant of unidirectional cooperation was investigated in \cite{SSKPS}, where the three encoders of a Gaussian MAC can cooperate over a ring of unidirectional links. However, only lower and upper bounds were found for the maximum achievable equal rate. 

Further literature exists for Willems conferencing on the decoding side of a multi-user network. For degraded discrete broadcast channels, the capacity region was found in \cite{DS} if the receivers can exchange information about the received codewords in a single conference step. For the general broadcast and multicast channels, achievability regions were determined. For the Gaussian relay channel, the dependence of the performance on the number of conferencing iterations between the receiver and the relay was investigated in \cite{NETAL}. For the Gaussian $Z$-interference channel, outer and inner bounds to the capacity region where the decoders can exchange information about the channel outputs are provided in \cite{DOS}. Finally, for discrete and Gaussian memoryless interference channels with conferencing decoders and where the senders have a common message, \cite{SGPGS} determines achievable regions. Exact capacity regions are determined if the channel is physically degraded. If the encoders can conference instead of having a common message, the situation is the same.

The discrete MAC with conferencing encoders is closely related to the discrete MAC with common message. Intuitively, the messages exchanged between the encoders in the cooperative setting form a common message, so the results known for the corresponding non-cooperative channel with common message can be applied to find the achievable rates of the cooperative setting. This transition was used in \cite{Wi1,Wi2,BLW,Wig}, and \cite{MYK}. The capacity region of the MAC with common message was determined in \cite{SW}, a simpler proof was found in \cite{Wi1}.

The goal of this paper is to generalize the original setting considered by Willems even further. We treat a compound discrete memoryless MAC with an arbitrary number of channel realizations. The receiver's CSI (CSIR) may be arbitrary between full and absent. The possible  transmitter's CSI (CSIT) may be different from CSIR and asymmetric at the two encoders. It is restricted to a finite number of instances, even though the number of actual channel realizations may be infinite. For this channel, we consider two cases. First, we characterize the capacity region of this channel where the transmitters have a common message. Then, we determine the capacity region of the channel where there is no common message any more. Instead, the encoders have access to the output of a rate-constrained noiseless two-user MAC. Each input node of the noiseless MAC corresponds to one of the transmitters of the compound MAC. Each input to the noiseless MAC consists of the pair formed by the message which is to be transmitted and the CSIT present at the corresponding transmitter. This generalizes Willems' conferencing to a non-causal conferencing protocol, where the conferencing capacities considered by Willems correspond to the rate constraints of the noiseless MAC in the generalized model. It turns out that this non-causal conferencing does not increase the capacity region, and as in \cite{Wi1,Wi2}, every rate contained in the capacity region can be achieved using a one-shot Willems ``conference''. We determine how large the conferencing capacities need to be in order to achieve the full-cooperation sum rate and the full-cooperation capacity region, respectively. The latter is particularly interesting because it shows that forming a ``virtual MIMO system'' as mentioned in Subsection \ref{subsect:motiv} and considered in \cite{KFV} does not require infinite cooperation capacity.

\subsection{Organization of the Paper}

In Section \ref{sect:model}, we address the problems presented above. We present the two basic channel models underlying our analysis: the compound MAC with common message and partial CSI and the compound MAC with conferencing encoders and partial CSI. We also introduce the generalized conferencing protocol used in the analysis of the conferencing MAC. We state the main results concerning the capacity regions of the two models. We also derive the minimal amount of cooperation needed in the conferencing setting in order to achieve the optimal (i.e. full-cooperation) sum rate and the optimal, full-cooperation rate region. The achievability of the rate regions claimed in the main theorems is shown in Section \ref{sect:ach}. The weak converses are shown in Section \ref{sect:conv}. Only the converse for the conferencing MAC is presented in detail, because the converse for the MAC with common message is similar to part of the converse for the MAC with conferencing encoders. We address the application of the MAC with conferencing encoders to the analysis of cellular systems where one data stream is split up and sent using different base stations in Section \ref{sect:numerik}. In the same section, in a simple numerical example, the capacity regions of a MAC with conferencing encoders is plotted for various amounts of cooperation. In the final section, we sum up the paper and discuss the directions of future research. In the Appendix several auxiliary lemmata concerning typical sequences are collected.

\subsection{Notation}

For real numbers $a$ and $b$, we set $a\wedge b:=\min(a,b)$ and $a\vee b:=\max(a,b)$.

For any positive integer $m$, we write $[1,m]$ for the set $\{1,\ldots,m\}$. The complement of a set $F\subset\X$ in $\X$ is denoted by $F^c$. The function $1_F$ is the indicator function of $F$, i.e. $1_F(x)$ equals 1 if $x\in F$ and 0 else.
For a set $E\subset\X\times\Y$, we write $E\rvert_y:=\{x\in\X:(x,y)\in\X\times\Y\}$. For a mapping $f:\X\rightarrow\Y$, define $\lVert f\rVert$ to be the cardinality of the range of $f$. 

Denote the set of probability measures on a discrete set $\X$ by $\P(\X)$. The $n$-fold product of a $p\in\P(\X)$ is denoted by $p^n\in\P(\X^n)$. By $\K(\Y\vert\X)$, we denote the set of stochastic matrices with rows indexed by $\X$ and columns indexed by $\Y$. The $n$-fold memoryless extension of a $W\in\K(\Y\vert\X)$ is defined as
\[
  W^n(\yy\vert\xx):=\prod_{m=1}^nW(y_m\vert x_m),
\]
where $\xx=(x_1,\ldots,x_n)\in\X^n,\yy=(y_1,\ldots,y_n)\in\Y^n$.

Let $\X$ be a finite set. For $\xx=(x_1,\ldots,x_n)\in\X^n$, define the type $p_\xx\in\P(\X)$ of $\xx$ by $np_\xx(x)=\lvert\{i:x_i=x\}\rvert$. For $\delta>0$ and $p\in\P(\X)$, define $T_{p,\delta}^n$ to be the set of those $\xx\in\X^n$ such that $\lvert p_\xx(x)-p(x)\rvert\leq\delta$ for all $x$ and such that $p_\xx(x)=0$ if $p(x)=0$.

\section{Channel Model and Main Results}\label{sect:model}

\subsection{The Channel Model}

Let $\X,\Y,\Z$ be finite sets. A compound discrete memoryless MAC with input alphabets $\X$ and $\Y$ and output alphabet $\Z$ is determined by a set of stochastic matrices $\W\subset\K(\Z\vert \X\times \Y)$. $\W$ may be finite of infinite. Every $W\in\W$ corresponds to a different channel state, so we will also call the elements $W$ the states of the compound MAC $\W$. The transmitter using alphabet $\X$ will be called transmitter (sender, encoder) 1 and the transmitter with alphabet $\Y$ will be called transmitter (sender, encoder) 2. If transmitter 1 sends a word $\xx=(x_1,\ldots,x_n)\in \X^n$ and transmitter 2 sends a word $\yy=(y_1,\ldots,y_n)\in \Y^n$, and if the channel state is $W\in\W$, then the receiver will receive the word $\zz=(z_1,\ldots,z_n)\in \Z^n$ with probability
\[
  W^n(\zz\vert\xx,\yy):=\prod_{m=1}^nW(z_m\vert x_m,y_m).
\]
The compound channel model does not include a change of state in the middle of a transmission block.

The goal is to find codes that are ``good'' (in a sense to be specified later) universally for all those channel states which might be the actual one according to CSI. In our setting, CSI at sender $\nu$ is given by a finite CSIT partition
\begin{equation}\label{CSIT}
  t_\nu=\{\W_{\tau_\nu}\subset\W:\tau_\nu\in T_\nu\}
\end{equation}
for $\nu=1,2$. The sets $T_1,T_2$ are finite, and the $\W_{\tau_\nu}$ satisfy
\[
  \bigcup_{\tau_\nu\in T_\nu}\W_{\tau_\nu}=\W,\qquad\text{and}\qquad\W_{\tau_\nu}\cap\W_{\tau_\nu'}=\varnothing\quad\text{if}\quad\tau_\nu\neq\tau_\nu'.
\]
Before encoding, transmitter $\nu$ knows which element of the partition the actual channel state is contained in, i.e. if $W\in\W_{\tau_\nu}$ is the channel state, then it knows $\tau_\nu$. With this knowledge, it can adjust its codebook to the channel conditions to some degree. For $\tau=(\tau_1,\tau_2)\in T_1\times T_2$, we denote by
\[
  \W_\tau:=\W_{\tau_1\tau_2}:=\W_{\tau_1}\cap\W_{\tau_2}
\]
the set of channel states which is possible according to the combined channel knowledge of both transmitters. Note that every function from $\W$ into a finite set induces a finite partition as in \eqref{CSIT}, so this is a very general concept of CSIT. At the receiver side, the knowledge about the channel state is given by a not necessarily finite CSIR partition
\begin{equation}\label{CSIR}
  r=\{\W_\rho\subset\W:\rho\in R\}.
\end{equation}
$R$ is an arbitrary set and the sets $\W_\rho$ satisfy
\[
  \bigcup_{\rho\in R}\W_\rho=\W\qquad\text{and}\qquad\W_\rho\cap\W_{\rho'}=\varnothing\quad\text{if}\quad\rho\neq\rho'.
\]
If the channel state is $W\in\W_\rho$, then the receiver knows $\rho$. Thus it can adjust its decision rule to this partial channel knowledge. This concept includes any kind of deterministic CSIR, because any function from $\W$ into an arbitrary set induces a partition as in \eqref{CSIR}. Note that if $\W$ is infinite, the transmitters can never have full CSI, whereas this is possible for the receiver if $r=\{\{W\}:W\in\W\}$.

\begin{defn}
  The compound discrete memoryless MAC $\W$ together with the CSIT partitions $t_1,t_2$ and the CSIR partition $r$ is denoted by the quadruple $(\W,t_1,t_2,r)$.
\end{defn}

\begin{ex}
There are several communication situations which are appropriately described by a compound MAC. One case is where information is to be sent from two transmitting terminals to one receiving terminal through a fading channel. If the channel remains constant during one transmission block, one obtains a compound channel. Usually, CSIT is not perfect. It might be, however, that the transmitters have access to partial CSI, e.g. by using feedback. This will not determine an exact channel state, but only an approximation. Coding must then be done in such a way that it is good for all those channel realizations which are possible according to CSIT. 

Another situation to be modeled by compound channels occurs if there are two transmitters each of which would like to send one message to several receivers at the same time. The channels to the different receivers differ from each other because all the terminals are at different locations. Now, the following meaning can be given to the above variants of channel knowledge. If CSIT is given as $\tau=(\tau_1,\tau_2)$, this describes that the information is not intended for all receivers, but only for those contained in $\W_\tau$. Knowledge about the intended receivers may be asymmetric at the senders. If every receiver has its own decoding procedure, full CSIR (i.e. $r=\{\{W\}:W\in\W\}$) would be a natural assumption. If the receivers must all use the same decoder, there is no CSIR. Non-trivial CSIR could mean that independently of the decision at the transmitters where data are to be sent (modeled by CSIT), a subset of receivers is chosen as the set which the data are intended for without informing the transmitters about this decision.
\end{ex}

\subsection{The MAC With Common Message}\label{subsect:CM}

Let the channel $(\W,t_1,t_2,r)$ be given. We now present the first of the problems treated in this paper, the capacity region of the compound MAC with common message. It is an interesting information-theoretic model in itself. However, its main interest, at least in this paper, is that it provides a basis for the solution of the problem presented in the next section, which is the capacity region of the compound MAC with conferencing encoders.

Assume that each transmitter has a set of private messages $[1,M_\nu]$, $\nu=1,2$, and that both transmitters have an additional set of common messages $[1,M_0]$ for the receiver (Fig. \ref{fig:cmsystem}). Let $n$ be a positive integer.

\begin{figure}
  \begin{center}
    \includegraphics[width=.7\linewidth]{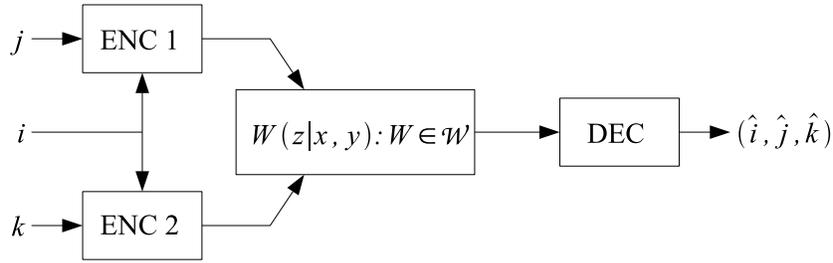}
    \caption{The MAC with Common Message}
    \label{fig:cmsystem}
  \end{center}
\end{figure}

\begin{defn}
A code$_\CM(n,M_0,M_1,M_2)$ is a triple $(f_1,f_2,\Phi)$ of functions satisfying
\begin{align*}
  f_1&:[1,M_0]\times[1,M_1]\times T_1\rightarrow\X^n,\\
  f_2&:[1,M_0]\times[1,M_2]\times T_2\rightarrow\Y^n,\\
  \Phi&:\Z^n\times R\rightarrow[1,M_0]\times[1,M_1]\times[1,M_2].
\end{align*}
$n$ is called the blocklength of the code.
\end{defn}

\begin{rem}\label{functdescrcm}
  Clearly, the codes$_\CM(n,M_0,M_1,M_2)$ are in one-to-one correspondence with the families
\begin{equation}\label{codecmform}
  \{(\x{ij}{\tau_1},\y{ik}{\tau_2},F_{ijk}^\rho):(i,j,k)\in[1,M_0]\times[1,M_1]\times[1,M_2], (\tau_1,\tau_2,\rho)\in T_1\times T_2\times R\},
\end{equation}
where $\x{ij}{\tau_1}\in\X^n$, $\y{ik}{\tau_2}\in\Y^n$, and where the $F_{ijk}^\rho\subset\Z^n$ satisfy
\[
  F_{ijk}^\rho\cap F_{i'j'k'}^\rho=\varnothing\quad\text{if}\quad(i,j,k)\neq(i',j',k').
\]
(The sets $F_{ijk}^\rho$ are obtained from $\Phi$ by setting 
\[
  F_{ijk}^\rho:=\{\zz\in\Z^n:\Phi(\zz,\rho)=(i,j,k)\}.\text{)}
\]
In the following, we will use the description of codes$_\CM$ as families as in \eqref{codecmform}. The functional description of codes will be of use when we are dealing with transmitter cooperation. We say more on that in Remark \ref{functdescr}.
\end{rem}
The $\x{ij}{\tau_1}$ and $\y{ik}{\tau_2}$ are the codewords and the $F_{ijk}^\rho$ are the decoding sets of the code. Let the transmitters have the common message $i$. Suppose that transmitter $1$ additionally has the private message $j$ and knows that $W\in\W_{\tau_1}$. Then it uses the codeword $\x{ij}{\tau_1}$. If transmitter $2$ additionally has the private message $k$ and knows that $W\in\W_{\tau_2}$, it uses the codeword $\y{ik}{\tau_2}$. Suppose that the receiver knows that $W\in\W_{\rho}$. If the channel output $\zz\in \Z^n$ is contained in $F_{ijk}^\rho$, the receiver decides that the message triple $(i,j,k)$ has been sent.

\begin{defn}
For $\lambda\in(0,1)$, a code$_\CM(n,M_0,M_1,M_2)$ is a code$_\CM(n,M_0,M_1,M_2,\lambda)$ if
\[
  \sup_{\tau_1,\tau_2,\rho}\;\sup_{W\in\W_{\tau_1\tau_2}\cap\W_\rho}\frac{1}{M_0M_1M_2}\sum_{i,j,k}W^n\bigl((F_{ijk}^\rho)^c\vert\x{ij}{\tau_1},\y{ik}{\tau_2}\bigr)\leq\lambda.
\]
\end{defn}
That means that for every instance of channel knowledge at the transmitters and at the receiver, the encoding/decoding chosen for this instance must yield a small average error for every channel state that may occur according to the CSI. In other words, the code chosen for a particular instance $(\tau_1,\tau_2,\rho)$ of CSI must be universally good for the class of channels $\{W\in\W_{\tau_1\tau_2}\cap\W_\rho\}$.

The first goal in this paper is to characterize the capacity region of the compound MAC with common message. That means that we will characterize the set of achievable rate triples and prove a weak converse.
\begin{defn}
A rate triple $(R_0,R_1,R_2)$ is \emph{achievable} for the compound channel $(\W,t_1,t_2,r)$ with common message if for every $\eps>0$ and $\lambda\in (0,1)$ and for $n$ large enough, there is a code$_\CM(n,M_0,M_1,M_2,\lambda)$ with
\[
  \frac{1}{n}\log M_\nu\geq R_\nu-\eps\qquad\text{for every }\nu=0,1,2.
\]
We denote the set of achievable rate triples by $\mathcal{C}_\CM(\W,t_1,t_2,r)$.
\end{defn}

Before stating the theorem on the capacity region, we need to introduce some new notation. We set $\Pi_1$ to be the set of families
\[
  p:=\{p_{\tau_1\tau_2}(u,x,y)=p_0(u)p_{1\tau_1}(x\vert u)p_{2\tau_2}(y\vert u):(\tau_1,\tau_2)\in T_1\times T_2\},
\]
of probability distributions, where $p_0$ is a distribution on a finite subset of the integers, and where $(p_{1\tau_1},p_{2\tau_2})\in\K(\X\vert\U)\times\K(\Y\vert\U)$ for each $(\tau_1,\tau_2)$. Every $p\in\Pi_1$ defines a family of probability measures on $\U\times\X\times\Y\times\Z$, where $\U$ is the set corresponding to $p$. This family consists of the probability measures $p_W$ ($W\in\W$), where
\begin{equation}\label{peta}
  p_W(u,x,y,z)=p_0(u)p_{1\tau_1}(x\vert u)p_{2\tau_2}(y\vert u)W(z\vert x,y),
\end{equation}
and where $(\tau_1,\tau_2)\in T_1\times T_2$ is such that $W\in\W_{\tau_1\tau_2}$. Let the quadruple of random variables $(U,X_{\tau_1},Y_{\tau_2},Z_W)$ take values in $\U\times \X\times \Y\times \Z$ with joint probability $p_W$. Then, define the set $\mathcal{R}_\CM(p,\tau_1,\tau_2,W)$ to be the set of $(R_0,R_1,R_2)$, where every $R_\nu\geq 0$ and where
\begin{align*}
          R_1&\leq I(Z_W;X_{\tau_1}\vert Y_{\tau_2},U),\\
          R_2&\leq I(Z_W;Y_{\tau_2}\vert X_{\tau_1},U),\\
      R_1+R_2&\leq I(Z_W;X_{\tau_1},Y_{\tau_2}\vert U),\\
  R_0+R_1+R_2&\leq I(Z_W;X_{\tau_1},Y_{\tau_2}).
\end{align*}
Defining 
\begin{align*}
  \mathcal{C}^*_\CM(\W,t_1,t_2):=\bigcup_{p\in\Pi_1}\;\bigcap_{(\tau_1,\tau_2)\in T_1\times T_2}\;\bigcap_{W\in\W_{\tau_1\tau_2}}\mathcal{R}_\CM(p,\tau_1,\tau_2,W),
\end{align*}
we are able to state the first main result.
\begin{theorem}\label{thmcm}
  For the compound MAC $(\W,t_1,t_2,r)$, one has
\begin{align*}
  \mathcal{C}_\CM(\W,t_1,t_2,r)=\mathcal{C}^*_\CM(\W,t_1,t_2),
\end{align*}
and there is a weak converse. More exactly, for every $(R_0,R_1,R_2)$ in $\mathcal{C}_\CM(\W,t_1,t_2,r)$ and for every $\eps>0$, there is a $\zeta$ such that there exists a sequence of codes$_\CM(n,M_0^ {(n)},M_1^{(n)},M_2^{(n)},2^{-n\zeta})$ fulfilling
\[
  \frac{1}{n}\log M_\nu^{(n)}\geq R_\nu-\eps,\quad\nu=0,1,2,
\]
if $n$ is large, i.e. one has exponential decay of the error probability with increasing blocklength.

$\mathcal{C}_\CM(\W,t_1,t_2,r)$ is convex. The cardinality of the auxiliary set $\U$ can be restricted to be at most $\min(\lvert \X\rvert\lvert \Y\rvert+2,\lvert \Z\rvert+3)$.
\end{theorem}
\begin{rem}\label{remcm}
  \begin{enumerate}
    \item A weak converse states that if a code has rates which are further than $\eps>0$ from the capacity region and if its blocklength is large, then the average error of this code must be larger than a constant only depending on $\eps$. A moment's thought reveals that this is a stronger statement than just saying that the rates outside of the capacity region are not achievable.\label{remcm0}
    \item $\mathcal{C}_\CM(\W,t_1,t_2,r)$ is independent of the CSIR partition $r$. That means that given a certain CSIT, the capacity region does not vary as CSIR varies. A heuristic explanation of this phenomenon is given in \cite[Section 4.5]{W} for the case of single-user compound channels. It builds on the fact that the receiver can estimate the channel from a pilot sequence with a length which is negligible compared to the blocklength.\label{remcm1}
    \item Note that first taking a union and then an intersection of sets in the definition of $\mathcal{C}^*_\CM(\W,t_1,t_2)$ is similar to the max-min capacity expression for the classical single-user discrete memoryless compound channel \cite{CK}. We write two intersections instead of one in order to make the difference clear which remains between the two expressions. Recall that the $p\in\Pi_1$ are \emph{families} of probability measures. Every choice $(\tau_1,\tau_2)\in T_1\times T_2$ activates a certain element of such a family $p$. The union and the first intersection are thus related in a more complex manner than in the single-user expression.
    \item As CSIT increases, the capacity region grows, and in principle, one can read off from this how the region scales with increasing channel knowledge at the transmitters. More precisely, assume that there are pairs $(t_1,t_2)$ and $(t_1',t_2')$ of CSIT partitions,
\[
  t_\nu=\{\W_{\tau_\nu}:\tau_\nu\in T_\nu\},\qquad t_\nu'=\{\W'_{\tau_\nu'}:\tau_\nu\in T_\nu'\}\qquad(\nu=1,2),
\]
 such that $t_\nu'$ is finer than $t_\nu$ ($\nu=1,2$). That means that for every $\W'_{\tau_\nu'}\in t_\nu'$ there is a $\tau_\nu\in T_\nu$ with $\W'_{\tau_\nu'}\subset\W_{\tau_\nu}$, so one can assume that $T_\nu\subset T_\nu'$. Observe that the $\Pi_1$ corresponding to $(t_1,t_2)$, which we call $\Pi_1(t_1,t_2)$ only in this remark, can naturally be considered a subset of $\Pi_1(t_1',t_2')$, which denotes the $\Pi_1$ corresponding to $(t_1',t_2')$ only for this remark. Thus
\[
  \bigcup_{p\in\Pi_1(t_1,t_2)}\;\bigcap_{(\tau_1,\tau_2)\in T_1\times T_2}\;\bigcap_{W\in\W_{\tau_1\tau_2}}\mathcal{R}_\CM(p,\tau_1,\tau_2,W)=\bigcup_{p\in\Pi_1(t_1,t_2)}\;\bigcap_{(\tau_1',\tau_2')\in T_1'\times T_2'}\;\bigcap_{W\in\W'_{\tau_1'\tau_2'}}\mathcal{R}_\CM(p,\tau_1',\tau_2',W),
\]
and it follows that $\mathcal{C}_\CM(\W,t_1,t_2,r)\subset\mathcal{C}(\W,t_1',t_2',r)$.
\label{remcm3}
  \end{enumerate}
\end{rem}

\subsection{The MAC with Conferencing Encoders}\label{subsect:CONF}

Again let the channel $(\W,t_1,t_2,r)$ be given. Here we assume that each transmitter only has a set of private messages $[1,M_\nu]$ ($\nu=1,2$) for the receiver. Encoding is done in three stages. In the first stage, each encoder transmits its message and CSIT to a central node, a ``switch'', over a noiseless rate-constrained discrete MAC. The rate constraints are part of the problem setting and thus fixed, but the noiseless MAC is not given, it is part of the code. For reasons that will become clear soon, we call it a ``conferencing MAC''. In the second stage, the information gathered by the switch is passed on to each encoder over channels without incurring noise or loss. The codewords are chosen in the third stage. Each encoder chooses its codewords using three parameters: the message it wants to transmit, its CSIT, and the output of the conferencing MAC. This is illustrated in Fig. \ref{fig:confsystem}. 

\begin{figure}
  \begin{center}
    \includegraphics[width=.7\linewidth]{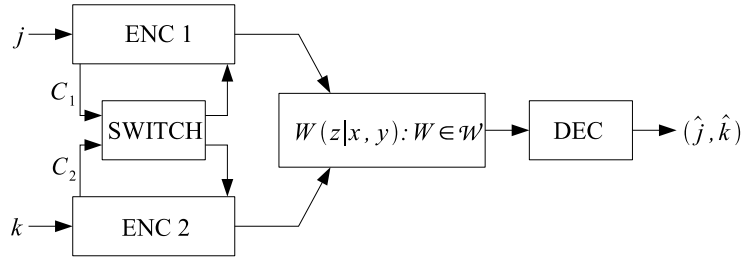}
    \caption{The MAC with Conferencing Encoders}
    \label{fig:confsystem}
  \end{center}
\end{figure}

The conferencing MAC can be chosen freely within the constraints, so it can be seen as a part of the encoding process. Assume that the blocklength of the codes used for transmission is set to be $n$. The rate constraints $(C_1,C_2)$ are such that $nC_\nu$ is the maximal number of bits transmitter $\nu$ can communicate to the receiving node of the conferencing MAC. Thus if transmitter 1, say, has message $j$ and CSIT $\tau_1$, then transmitter 2, who knows neither $j$ nor $\tau_1$, can use at most $C_1$ additional bits from transmitter 1 to encode its own message. Consequently, there is a limited degree of cooperation between the encoders enhancing the reliability of transmission. As the constraints on the noiseless MAC are measured in terms of $n$, one can interpret the communication over this channel as taking place during the transmission over $(\W,t_1,t_2,r)$ of the codeword preceding that which is constructed with the help of the conferencing MAC. 

Example \ref{iter} below shows how this kind of coding generalizes coding using Willems conferencing functions as defined in \cite{Wi1,Wi2}. From Theorem \ref{thmconf} below it follows that Willems conferencing is more than just a special case. In fact, it suffices to achieve the capacity region. In Section \ref{subsect:numapp}, we give an application where it is useful to have the more general notion of conferencing which is used here.

We now come to the formal definitions. Recall that a noiseless MAC is nothing but a function from a Cartesian product to some other space. 

\begin{defn}\label{codeconf}
A code$_\CONF(n,M_1,M_2,C_1,C_2)$ is a quadruple $(f_1,f_2,g,\Phi)$ of functions which satisfy
\begin{align*}
  f_1&:[1,M_1]\times\Gamma\times T_1\rightarrow\X^n,\\
  f_2&:[1,M_2]\times\Gamma\times T_2\rightarrow\Y^n,\\
  g&:[1,M_1]\times[1,M_2]\times T_1\times T_2\rightarrow\Gamma,\\
  \Phi&:\Z^n\times R\rightarrow[1,M_1]\times[1,M_2],
\end{align*}
where $\Gamma$ is a finite set and where $g$ satisfies
\begin{align}
  \frac{1}{n}\log\lVert g_{(j,\tau_1)}\rVert\leq C_2\qquad\text{for all }(j,\tau_1)\in[1,M_1]\times T_1,\label{g1}\\
  \frac{1}{n}\log\lVert g_{(k,\tau_2)}\rVert\leq C_1\qquad\text{for all }(k,\tau_2)\in[1,M_1]\times T_2\label{g2}
\end{align}
for the functions $g_{(j,\tau_1)}$ and $g_{(k,\tau_2)}$ defined by $g_{(j,\tau_1)}(j,k,\tau_1,\tau_2)=g_{(k,\tau_2)}(j,k,\tau_1,\tau_2)=g(j,k,\tau_1,\tau_2)$. The number $n$ is called the blocklength of the code. $g$ is called a conferencing MAC or alternatively a generalized conferencing function. The latter name is justified by Example \ref{iter}.
\end{defn}

\begin{rem}\label{functdescr}
  Analogous to the situation for the MAC with common message described in Remark \ref{functdescrcm}, the code$_\CONF$ $(n,M_1,M_2,C_1,C_2)$ given by the quadruple $(f_1,f_2,g,\Phi)$ uniquely determines a family
\begin{equation}\label{codeseq}
  \{(\x{jk}{\tau_1\tau_2},\y{jk}{\tau_1\tau_2},F_{jk}^\rho):(j,k)\in[1,M_1]\times[1,M_2],(\tau_1,\tau_2,\rho)\in T_1\times T_2\times R\}.
\end{equation}
For the elements of this family, $\x{jk}{\tau_1\tau_2}\in\X^n$ (not necessarily different!), $\y{jk}{\tau_1\tau_2}\in\Y^n$ (not necessarily different!), and the $F_{jk}^\rho\subset\Z^n$ satisfy
\[
  F_{jk}^\rho\cap F_{j'k'}^\rho=\varnothing\qquad\text{if }(j,k)\neq (j',k').
\]
For every $(\tau_1,\tau_2)\in T_1\times T_2$, the family \eqref{codeseq} must satisfy
\begin{align}
  \x{jk}{\tau_1\tau_2}&=\x{jk'}{\tau_1\tau_2'}&&\qquad\text{if}\quad g(j,k,\tau_1,\tau_2)=g(j,k',\tau_1,\tau_2'),\label{xe}\\
  \y{jk}{\tau_1\tau_2}&=\y{j'k}{\tau_1'\tau_2}&&\qquad\text{if}\quad g(j,k,\tau_1,\tau_2)=g(j',k,\tau_1',\tau_2).\label{ys}
\end{align}
Thus an alternative definition of codes$_\CONF$ would be families like the family \eqref{codeseq} together with conferencing MACs as in \eqref{g1} and \eqref{g2}. This is the form we will mostly use in the paper because of shorter notation. However, the original definition \ref{codeconf} is  more constructive and gives more insights into the practical use of such codes. It will be used in the converse, where the way how the codewords depend on the messages will be exploited.
\end{rem}

\begin{rem}\label{rem:detmac}
  Note that \eqref{g1} and \eqref{g2} really are rate constraints. Indeed, let $(S_1,S_2)$ be a rate triple achievable by the MAC defined by $g$, where the average error criterion is used\footnote{Even though the channel is noiseless, this does make a difference. In fact, Dueck showed in \cite{D} that the maximal and the average error criteria differ for MACs using the example of a \emph{noiseless} channel!}. Then by the characterization of the MAC with non-cooperating encoders without common message (cf. \cite[Theorem 3.2.3]{CK}), there must be independent random variables $J$ on $[1,M_1]\times T_1$ and $K$ on $[1,M_2]\times T_2$ such that 
\begin{align}
  S_1&\leq I(g(J,K);J\vert K)=H(g(J,K)\vert K),\label{S_1}\\
  S_2&\leq I(g(J,K);K\vert J)=H(g(J,K)\vert J),\label{S_2}\\
  S_1+S_2&\leq I(g(J,K);J,K)=H(g(J,K)).\label{S_1S_2}
\end{align}
  But by the constraints \eqref{g1} and \eqref{g2}, one knows that the right side of \eqref{S_1} must be smaller than $nC_1$and the right side of \eqref{S_2} must be smaller than $nC_2$. Clearly, the sum rate then must be smaller than $n(C_1+C_2)$. Moreover, as the bounds in \eqref{S_1}-\eqref{S_1S_2} are achievable, it even follows $H(g(J,K))\leq n(C_1+C_2)$ for every admissible choice of $J$ and $K$.
\end{rem}

With the above definition, the coding scheme is obvious: if the message pair $(j,k)$ is to be transmitted and if the pair of CSIT instances is $(\tau_1,\tau_2)$, then the senders use the codewords $\x{jk}{\tau_1,\tau_2}$ and $\y{jk}{\tau_1,\tau_2}$, respectively. If CSIR is $\rho$ and if the channel output is contained in the decoding set $F_{jk}^\rho$, then the receiver decides that the message pair $(j,k)$ has been transmitted.

\begin{defn}
For $\lambda\in(0,1)$, a code$_\CONF(n,M_1,M_2,C_1,C_2)$ is a code$_\CONF(n,M_1,M_2,C_1,C_2,\lambda)$ if
\[
  \sup_{\tau_1,\tau_2,\rho}\;\sup_{W\in\W_{\tau_1\tau_2}\cap\W_\rho}\frac{1}{M_1M_2}\sum_{j,k}W^n\bigl((F_{jk}^\rho)^c\vert\x{jk}{\tau_1,\tau_2},\y{jk}{\tau_1,\tau_2}\bigr)\leq\lambda.
\]
\end{defn}

In the following example, we prove our claim that using generalized conferencing in the encoding process generalizes Willems' conferencing encoders. We fix the notation
\begin{equation}\label{vbar}
  \bar\nu:=\begin{cases}
             1\quad\text{if }\nu=2,\\
             2\quad\text{if }\nu=1.
           \end{cases}
\end{equation}

\begin{ex}[Willems Conferencing Functions]\label{iter}
  Let positive integers $V_1$ and $V_2$ be given which can be written as products
\[
  V_\nu=V_{\nu,1}\cdots V_{\nu,I}
\]
for some positive integer $I$ which does not depend on $\nu$. Assume that 
\[
  \frac{1}{n}\log V_\nu\leq C_\nu.
\]
We first give a formal definition of a pair of Willems conferencing functions $(g_1,g_2)$. Such a pair is determined in an iterative manner via sequences of functions $h_{1,1},\ldots,h_{1,I}$ and $h_{2,1},\ldots,h_{2,I}$, where for $\nu=1,2$ and $i=2,\ldots,I$,
\begin{align*}
  h_{\nu,1}&:[1,M_\nu]\times T_\nu\rightarrow[1,V_{\nu,1}],\\
  h_{\nu,i}&:[1,M_\nu]\times T_\nu\times[1,V_{\bar\nu,1}]\times\ldots\times[1,V_{\bar\nu,i-1}]\rightarrow[1,V_{\nu,i}].
\end{align*}
For $\nu=1,2$ and $i=2,\ldots,I$, one recursively defines functions 
\begin{align*}
  h_{\nu,1}^*&:[1,M_\nu]\times T_\nu\rightarrow[1,V_{\nu,1}],\\
  h_{\nu,i}^*&:[1,M_1]\times[1,M_2]\times T_1\times T_2\rightarrow[1,V_{\nu,i}]
\end{align*}
by
\begin{align*}
  h_{\nu,1}^*(\ell_\nu,\tau_\nu)&=h_{\nu,1}(\ell_\nu,\tau_\nu),\\
  h_{\nu,i}^*(\ell_1,\ell_2,\tau_1,\tau_2)&=h_{\nu,i}\bigl(\ell_\nu,\tau_\nu,h_{\bar\nu,1}^*(\ell_{\bar\nu},\tau_{\bar\nu}),\ldots,h_{\bar\nu,i-1}^*(\ell_1,\ell_2,\tau_1,\tau_2)\bigr).
\end{align*}
The functions $g_1,g_2$ are then obtained by setting 
\[
  g_\nu:=(h_{\nu,1}^*,\ldots,h_{\nu,I}^*).
\]
One checks easily that $g=(g_1,g_2)$ is a noiseless MAC with output alphabet $\Gamma=[1,V_1]\times[1,V_2]$ satisfying \eqref{g1} and \eqref{g2}. Clearly, $g_\nu(\ell_1,\ell_2,\tau_1,\tau_2)$ is known at transmitter $\nu$ because it only depends on $(\ell_\nu,\tau_\nu)$ and $g_{\bar\nu}(\ell_1,\ell_2,\tau_1,\tau_2)$.

Note that not every conferencing MAC $g=(g_1,g_2)$ with output alphabet $[1,V_1]\times[1,V_2]$ can be obtained through Willems conferencing. The most trivial example to see this is where $V_1$ is prime and where the conferencing function $g_1$ mapping into $[1,V_1]$ depends on $k$. However, this setting can be given an interpretation in terms of MACs. Every pair of Willems' conferencing functions is nothing but the $I$-fold use of a non-stationary noiseless MAC with feedback. The above description of a transmission block of length $I$ over such a ``Willems channel'' as the one-shot use of a noiseless MAC as above is possible because noise plays no role here.
\end{ex}

For achievability and weak converse, we adapt the definitions from \ref{subsect:CM} to the conferencing setting. Let $C_1,C_2$ be nonnegative real numbers at least one of which is strictly greater than 0.
\begin{defn}
A rate pair $(R_1,R_2)$ is \emph{achievable} for the compound channel $(\W,t_1,t_2,r)$ with conferencing encoders with conferencing capacities $(C_1,C_2)$ if for every $\eps>0$ and $\lambda\in (0,1)$ and for $n$ large enough, there is a code$_\CONF(n,M_1,M_2,C_1,C_2,\lambda)$ with
\[
  \frac{1}{n}\log M_\nu\geq R_\nu-\eps
\]
We denote the set of achievable rate pairs by $\mathcal{C}_\CONF(\W,t_1,t_2,r,C_1,C_2)$.
\end{defn}

To state the result, we need to define the sets $\mathcal{R}_\CONF$. We denote by $\Pi_2$ the set of families 
\[
  p=\{p_{\tau_1\tau_2}(u,x,y)=p_0(u)p_{1\tau_1\tau_2}(x\vert u)p_{2\tau_1\tau_2}(y\vert u):(\tau_1,\tau_2)\in T_1\times T_2\}
\]
of probability distributions, where $p_0$ is a distribution on a finite subset $\U$ of the integers and where $(p_{1\tau_1\tau_2},p_{2\tau_1\tau_2})\in\K(\X\vert \U)\times\K(\Y\vert \U)$ for every $(\tau_1,\tau_2)\in T_1\times T_2$ (cf. the definition of $\Pi_1$ in Subsection \ref{subsect:CM}). Every $p\in\Pi_2$ defines a family of probability measures $p_W$ ($W\in\W$) on $\U\times\X\times\Y\times\Z$, where $\U$ is the set corresponding to $p$. This family consists of the probability measures $p_W$ ($W\in\W$) defined by
\[
  p_W(u,x,y,z):=p_0(u)p_{1\tau_1\tau_2}(x\vert u)p_{2\tau_1\tau_2}(y\vert u)W(z\vert x,y),
\]
where $(\tau_1,\tau_2)\in T_1\times T_2$ is such that $W\in\W_{\tau_1\tau_2}$. Finally we define subsets $\Pi_3$ and $\Pi_4$ of $\Pi_2$. $\Pi_3$ consists of those $p\in\Pi_2$ where the $p_{1\tau_1\tau_2}$ do not depend on $\tau_2$ and $\Pi_4$ consists of those $p\in\Pi_2$ where the $p_{2\tau_1\tau_2}$ do not depend on $\tau_1$.

For $W\in\W_\tau$, let $(U,X_\tau,Y_\tau,Z_W)$ be a quadruple of random variables which is distributed according to $p_W$. The set $\mathcal{R}_\CONF(p,W,C_1,C_2)$ is defined as the set of those pairs $(R_1,R_2)$ of non-negative reals which satisfy
\begin{align*}
      R_1&\leq I(Z_W;X_\tau\vert Y_\tau,U)+C_1,\\
      R_2&\leq I(Z_W;Y_\tau\vert X_\tau,U)+C_2,\\
  R_1+R_2&\leq\bigl(I(Z_W;X_\tau,Y_\tau\vert U)+C_1+C_2\bigr)\wedge I(Z_W;X_\tau,Y_\tau).
\end{align*}
If $C_1,C_2>0$, define the set 
\begin{align*}
  \mathcal{C}^*_\CONF(\W,t_1,t_2,C_1,C_2):=\bigcup_{p\in\Pi_2}\;\bigcap_{(\tau_1,\tau_2)\in T_1\times T_2}\;\bigcap_{W\in\W_{\tau_1\tau_2}}\mathcal{R}_\CONF(p,W,C_1,C_2).
\end{align*}
If $C_1>0,C_2=0$ (the reverse case is analogous with $\Pi_4$ replacing $\Pi_3$), define the set
\begin{align*}
  \mathcal{C}^*_{\CONF,1}(\W,t_1,t_2,C_1):=\bigcup_{p\in\Pi_3}\;\bigcap_{(\tau_1,\tau_2)\in T_1\times T_2}\;\bigcap_{W\in\W_{\tau_1\tau_2}}\mathcal{R}_\CONF(p,W,C_1,0).
\end{align*}

\begin{theorem}\label{thmconf}
  For the channel $(\W,t_1,t_2,r)$ and the pair $(C_1,C_2)$ of nonnegative real numbers, one has
\begin{align*}
  \mathcal{C}_\CONF(\W,t_1,t_2,r,C_1,C_2)=\begin{cases}
                               \mathcal{C}_\CONF^*(\W,t_1,t_2,C_1,C_2)&\quad\text{if }C_1,C_2>0,\\
                               \mathcal{C}_{\CONF,1}^*(\W,t_1,t_2,C_1)    &\quad\text{if }C_1>0,C_2=0,\\
                               \mathcal{C}_{\CONF,2}^*(\W,t_1,t_2,C_2)    &\quad\text{if }C_1=0,C_2>0.
                             \end{cases}
\end{align*}
This set can already be achieved using one-shot Willems conferencing functions, i.e. functions as defined in Example \ref{iter} with $I=1$.
More exactly, for every $(R_1,R_2)\in\mathcal{C}_\CONF(\W,t_1,t_2,r,C_1,C_2)$ and for every $\eps>0$, there is a $\zeta$ such that there exists a sequence of codes$_\CONF(n,M_1^{(n)},M_2^{(n)},C_1,C_2,2^{-n\zeta})$ fulfilling
\[
  \frac{1}{n}\log M_\nu^{(n)}\geq R_\nu-\eps,\quad\nu=1,2
\]
for large $n$ and using a one-shot Willems conference. $\mathcal{C}_\CONF(\W,t_1,t_2,r,C_1,C_2)$ is convex. One also has a weak converse. Further, the cardinality the auxiliary set $\U$ can be restricted to be at most $\min(\lvert \X\rvert\lvert \Y\rvert+2,\lvert \Z\rvert+3)$.
\end{theorem}

Remark \ref{remcm} applies here, too. Further, we note
\begin{rem}
  If $C_1,C_2>0$, then $\mathcal{C}_\CONF^*(\W,t_1,t_2,C_1,C_2)=\mathcal{C}_\CONF^*(\W,t,t,C_1,C_2)$, where
\[
  t=\{\W_{\tau_1\tau_2}:(\tau_1,\tau_2)\in T_1\times T_2\}.
\]
Thus bidirectional conferencing leads to a complete exchange of CSIT. The capacity region only depends on the joint CSIT at both transmitters, the asymmetry is lost.
\end{rem}
 Before beginning with the proof in the next section, we use Theorem \ref{thmconf} to find out how much cooperation is necessary to achieve the full-cooperation performance, i.e. the performance achieved when $C_1=C_2=\infty$, if cooperation in both directions is possible at all. (So we do not ask how large $C_1$ must be if $C_2=0$.) By Theorem \ref{thmconf}, the region of rates achievable with full cooperation is given by 
\begin{equation}\label{reform-A}
  0\leq R_1+R_2\leq C^\infty:=\max_{p\in\Pi_2}\;\min_{\tau\in T_1\times T_2}\;\inf_{W\in\W_\tau}I(Z_\eta;X_\tau,Y_\tau).
\end{equation}
$C^\infty$ also determines the maximally achievable sum rate. 

Let $\mathcal{M}$ be the set of those $p\in\Pi_2$ which achieve the maximum in $\eqref{reform-A}$. Then
\begin{cor}\label{sumcapcor}
\begin{enumerate}
  \item The infinite cooperation sum capacity is achievable if and only if
\begin{equation}\label{c1c2}
  C_1+C_2\geq C^\infty-\max_\mathcal{M}\;\min_{\tau\in T_1\times T_2}\;\inf_{W\in\W_\tau}I(Z_W;X_\tau,Y_\tau).
\end{equation}
\item The full cooperation region is achieved if 
\begin{align*}
  C_1&\geq C^\infty-\max_{p\in\Pi_2}\;\min_{\tau\in T_1\times T_2}\;\inf_{W\in\W_\tau}I(Z_W;X_\tau\vert Y_\tau,U),\label{lini}\\
  C_2&\geq C^\infty-\max_{p\in\Pi_2}\;\min_{\tau\in T_1\times T_2}\;\inf_{W\in \W_\tau}I(Z_W;Y_\tau\vert X_\tau,U).\notag
\end{align*}
\end{enumerate}
In particular, infinite-capacity cooperation is neither necessary in order to achieve the full-cooperation sum rate nor to achieve the full-cooperation rate region.
\end{cor}

\begin{proof}
  1) Denote the maximal sum rate achievable with cooperation capacities $C_1,C_2>0$ by $C(C_1,C_2)$. As for $C^\infty$, the problem of finding $C(C_1,C_2)$ is a maximization problem: one has
\begin{align*}
  C(C_1,C_2)=\max_{p\in\Pi_2}\;\min_{\tau\in T_1\times T_2}\;\inf_{\eta\in \W_\tau}(I(Z_\eta;X_{\tau},Y_{\tau}\vert U)+C_1+C_2)\wedge I(Z_\eta;X_{\tau},Y_{\tau}).
\end{align*}
The equation
\begin{equation}\label{C(C_1,C_2)geqCinfty}
  C(C_1,C_2)\geq C^\infty
\end{equation}
holds if and only if there is a $p\in\Pi_2$ such that
\begin{align*}
  \min_{\tau\in T_1\times T_2}\;\inf_{W\in \W_\tau}(I(Z_W;X_{\tau},Y_{\tau}\vert U)+C_1+C_2)\wedge I(Z_W;X_{\tau},Y_{\tau})\geq C^\infty.
\end{align*}
That means in particular that 
\[
  \min_{\tau\in T_1\times T_2}\;\inf_{W\in \W_\tau} I(Z_W;X_{\tau},Y_{\tau})\geq C^\infty,
\]
so $p$ must maximize 
\[
  \min_{\tau}\;\inf_{W\in \W_\tau}I(Z_W;X_\tau,Y_\tau).
\]
Then \eqref{C(C_1,C_2)geqCinfty} is equivalent to
\begin{align*}
  \max_\mathcal{M}\;\min_{\tau\in T_1\times T_2}\;\inf_{W\in \W_\tau}(I(Z_W;X_{\tau},Y_{\tau}\vert U)+C_1+C_2)\geq C^\infty,
\end{align*}
and this proves \eqref{c1c2}.

2) This part is trivial.
\end{proof}

In Section \ref{sect:numerik}, we present a numerical example which shows how the rate region changes with the conferencing capacities.

\section{The Achievability Proofs}\label{sect:ach}

\subsection{The MAC with Common Message}\label{subsect:CMach}

The proof of the achievability of $\mathcal{C}_\CM^*(\W,t_1,t_2)$ proceeds as follows. We first show that $\mathcal{C}_\CM^*(\W,t_1,t_2)$ is achievable using random codes, where codewords and decoding sets are chosen at random and the error is measured by taking the mean average error over all realizations. For this part, we adapt the nice proof used by Jahn \cite{J} in the context of arbitrarily varying multiuser channels to the setting of the compound MAC with common message. It uses some hypergraph terminology. An alternative proof proceeding as in standard random coding can be found in \cite{WBB}. It uses the same encoding and the same decoding, but needs the additional assumption that $\lvert \W\rvert<\infty$. Next, we derandomize, i.e. we extract a good deterministic code from the random one. This is much easier than for arbitrarily varying channels. It is first done for $\lvert \W\rvert<\infty$, and then an approximation argument is used for the case $\lvert \W\rvert=\infty$.

We assume here that the receiver has no CSI and show that $\mathcal{C}_\CM^*(\W,t_1,t_2)$ is achievable. This gives an inner bound to the capacity region for arbitrary CSIR-function $r$. As $\rho$ is trivial in the no-CSIR case, we omit it in the notation.

\subsubsection{Hypergraphs}

A \emph{cubic hypergraph} is a discrete set of the form $\U\times\X\times\Y$ with a collection $\mathcal{E}$ of subsets $E\subset\U\times\X\times\Y$.

\begin{defn}
  Consider a family $\{(U_i,X_{ij},Y_{ik}):i\in[1,M_0],\;j\in[1,M_1],\;k\in[1,M_2]\}$ of random vectors, where the $U_i$ take values in $\U$, the $X_{ij}$ take values in $\X$, and the $Y_{ik}$ take values in $\Y$. This family is a random $(M_0,M_1,M_2)$-half lattice in $\U\times \X\times \Y$ if the family 
\[
  \bigl\{\{(U_i,X_{ij},Y_{ik}):(j,k)\in[1,M_1]\times[1,M_2]\}:i\in[1,M_0]\bigr\}
\]
of random vectors is i.i.d. and such that given $U_i$,
\begin{itemize}
  \item the pair of families $\{X_{ij}:j\in[1,M_1]\}$, $\{Y_{ik}:k\in[1,M_2]\}$ is conditionally independent,
  \item the family $X_{ij}$, where $j\in[1,M_1]$, is conditionally i.i.d,
  \item the family $Y_{ik}$, where $k\in[1,M_2]$, is conditionally i.i.d.
\end{itemize}
\end{defn}

Let a random $(M_0,M_1,M_2)$-half lattice on $\U\times \X\times \Y$ be realized on a probability space $(\Omega,\mathcal{F},\PP)$. For any $E\in\mathcal{E}$, $(i,j,k)\in[1,M_0]\times[1,M_1]\times[1,M_2]$, and $(u,x,y)\in \U\times \X\times \Y$, we define\footnote{Recall the notation defined in the Introduction.}
\begin{align*}
  P_E(i,j,k)&:=\PP\bigl[E\cap\menge{(U_{i'},X_{i'j'},Y_{i'k'}):i'\neq i,j',k'}
  \neq\varnothing\vert(U_i,X_{ij},Y_{ik})=(u,x,y)\bigr],\\
  P_{E\vert u}(i,j,k)&:=\PP\bigl[E\rvert_u\cap\menge{(X_{ij'},Y_{ik'}):j'\neq j,k'\neq k}
  \neq\varnothing\vert(U_i,X_{ij},Y_{ik})=(u,x,y)\bigr],\\
  P_{E\vert(u,x)}(i,j,k)&:=\PP\bigl[E\rvert_{(u,x)}\cap\menge{Y_{ik'}:k'\neq k}
  \neq\varnothing\vert(U_i,X_{ij},Y_{ik})=(u,x,y)\bigr],\\
  P_{E\vert(u,y)}(i,j,k)&:=\PP\bigr[E\rvert_{(u,y)}\cap\menge{X_{ij'}:j'\neq j}
  \neq\varnothing\vert(U_i,X_{ij},Y_{ik})=(u,x,y)\bigr].
\end{align*}
We now state an analogue to the Hit Lemmas in \cite{J} which, just like those, is proved immediately using the independence/conditional independence properties of the random $(M_0,M_1,M_2)$-half lattice and the union bound.
\begin{lem}\label{hitlemma}
For a random $(M_0,M_1,M_2)$-half lattice on $\U\times \X\times \Y$, and for any $E\in\mathcal{E}$, $(i,j,k)\in[1,M_0]\times[1,M_1]\times[1,M_2]$, and $(u,x,y)\in \U\times \X\times \Y$,
  \begin{align*}
    P_E(i,j,k)&\leq M_0M_1M_2\PP\eckig{(U_1,X_{11},Y_{11})\in E},\\
    P_{E\vert u}(i,j,k)&\leq M_1M_2\PP\eckig{\left.(X_{11},Y_{11})\in E\rvert_u\right\vert U_1=u},\\
    P_{E\vert(u,y)}(i,j,k)&\leq M_1\PP\eckig{\left.X_{11}\in E\rvert_{(u,y)}\right\vert U_1=u},\\
    P_{E\vert(u,x)}(i,j,k)&\leq M_2\PP\eckig{\left.Y_{11}\in E\rvert_{(u,x)}\right\vert U_1=u}.
  \end{align*}
Hence, for any probability measure $p$ on $(\U\times\X\times\Y)\times\mathcal{E}$,
\begin{equation*}
\begin{aligned}
  &\sum_{u,x,y,E}p(u,x,y,E)\bigl(P_E(i,j,k)\\&\quad+P_{E\vert u}(i,j,k)+P_{E\vert(u,x)}(i,j,k)+P_{E\vert(u,y)}(i,j,k)\bigr)\\
  &\leq M_0M_1M_2\max_{E}\PP\eckig{(U_1,X_{11},Y_{11})\in E}\\
  &\quad+M_1M_2\max_{E,u}\PP\eckig{\left.(X_{11},Y_{11})\in E\rvert_u\right\vert U_1=u}\\
  &\quad+M_1\max_{E,u,y}\PP\eckig{\left.X_{11}\in E\rvert_{(u,y)}\right\vert U_1=u}\\
  &\quad+M_2\max_{E,u,x}\PP\eckig{\left.Y_{11}\in E\rvert_{(u,x)}\right\vert U_1=u}.
\end{aligned}
\end{equation*}
\end{lem}

\subsubsection{The Encoding/Decoding Procedure}

We can now return to the proof of the achievability part of Theorem \ref{thmcm}. Let the channel $(\W,t_1,t_2)$ be given (recall that the receiver is assumed to have no CSIR). We define a random code with block length $n$ which encodes $M_0$ common messages, $M_1$ messages of the first transmitter, and $M_2$ messages of the second transmitter. The randomness of the code can be viewed in two ways. First, one can see it as a method of proof which allows us to find a number of codes from which we will select a good one later. However, the randomness could also be incorporated into the system. Given that the transmitters and the receiver have access to the common randomness needed in the definition of the code, this already gives an achievable rate region if this randomness is exploited in the coding process. During the proof, one will see that this region even is achievable using a maximal error criterion. One needs to use the average error criterion when the achievability proof for random codes is strengthened in order to obtain the desired achievability part of Theorem \ref{thmcm} which requires the use of deterministic codes.

Using the notation introduced before Theorem \ref{thmcm}, we define an i.i.d. set of $M_0$ i.i.d. families of random variables $\{(U_i,X_{ij}^{\tau_1},Y_{ik}^{\tau_2}):(j,k)\in[1,M_1]\times[1,M_2],(\tau_1,\tau_2)\in T_1\times T_2\}$. Let 
\[
  p=\{p_0(\cdot)p_{1\tau_1}(\cdot\vert\cdot)p_{2\tau_2}(\cdot\vert\cdot):(\tau_1,\tau_2)\in T_1\times T_2\}\in\Pi_1
\]
and let $\U$ be the corresponding finite subset of the integers. The distribution $p_0^n$ of each $U_i$ on $\U^n$ is the $n$-fold product of $p_0$. Given $U_i$, the rest of the random variables in family $i$ is assumed to be conditionally independent given $U_i$. The conditional distribution $p_{1\tau_1}^n$ of each $X_{ij}^{\tau_1}$ given $U_i$ on $\X^n$ is the $n$-fold memoryless extension of $p_{1\tau_1}$, and the conditional distribution $p_{2\tau_2}^n$ of each $Y_{ik}^{\tau_2}$ given $U_i$ on $\Y^n$ is the $n$-fold memoryless extension of $p_{2\tau_2}$. Given a message triple $(i,j,k)$ that is to be transmitted and given an instance $(\tau_1,\tau_2)$ of CSIT, the transmitters use the random codewords $X_{ij}^{\tau_1}$ and $Y_{ik}^{\tau_2}$. 

We now define the decoding procedure, which requires access to the same random experiment as used for encoding. Fix a $\delta>0$. The $p$ used in the encoding process and every $W\in\W$ define a probability measure $p_W$ as in \eqref{peta}. For every $\tau=(\tau_1,\tau_2)$, define a set 
\[
  E^\tau:=\bigcup_{W\in \W_\tau}T_{p_W,\delta}^n
\]
(cf. the notation section in the Introduction). This set does not depend on the state $W\in \W_\tau$. The decoding sets are defined as follows: $F_{ijk}$ consists exactly of those $\zz\in\Z^n$ which satisfy both of the following conditions:
\begin{itemize}
  \item there is a $(\tau_1,\tau_2)$ such that 
\[
  (U_i,X_{ij}^{\tau_1},Y_{ik}^{\tau_2})\in E^{\tau}\rvert_\zz,
\]
  \item for all $(i',j',k')\neq (i,j,k)$ and for all $(\tau_1,\tau_2)$,
\[
  (U_{i'},X_{i'j'}^{\tau_1},Y_{i'k'}^{\tau_2})\notin E^{\tau}\rvert_\zz.
\]
\end{itemize}
Clearly the $F_{ijk}$ are disjoint. This decision rule does not depend on $\tau$, nor on $W$.

\subsubsection{Bounding the Mean Maximal Error for Random Coding}

We now bound the mean maximal error incurred by random coding, i.e. for each message triple $(i,j,k)$, CSIT instance $\tau=(\tau_1,\tau_2)$, and channel state $W\in \W_\tau$, we ask how large 
\begin{equation}\label{meanmax}
  \EE\eckig{W^n(F_{ijk}^c\vert X_{ij}^{\tau_1},Y_{ik}^{\tau_2})}
\end{equation}
can be. The receiver makes an error (decides incorrectly) if for the channel output $\zz$, one of the following holds:
\begin{enumerate}[E1)]
  \item $(U_i,X_{ij}^{\tau_1'},Y_{ik}^{\tau_2'})\notin E^{\tau'}\vert_\zz$ for all $\tau'=(\tau_1',\tau_2')$,
  \item\label{i} there is an $i'\neq i$ and arbitrary $(j',k')$ and $\tau'=(\tau_1',\tau_2')$ such that
\[
  (U_{i'},X_{i'j'}^{\tau_1'},Y_{i'k'}^{\tau_2'})\in E^{\tau'}\vert_\zz,
\]
  \item\label{jk} there is a $j'\neq j$ and a $k'\neq k$ and arbitrary $\tau'=(\tau_1',\tau_2')$ such that
\[
  (U_{i},X_{ij'}^{\tau_1'},Y_{ik'}^{\tau_2'})\in E^{\tau'}\vert_\zz,
\]
  \item\label{j} there is a $j'\neq j$ and arbitrary $\tau'=(\tau_1',\tau_2')$ such that
\[
  (U_{i},X_{ij'}^{\tau_1'},Y_{ik}^{\tau_2'})\in E^{\tau'}\vert_\zz,
\]
  \item\label{k} there is a $k'\neq k$ and arbitrary $\tau'=(\tau_1',\tau_2')$ such that
\[
  (U_{i},X_{ij}^{\tau_1'},Y_{ik'}^{\tau_2'})\in E^{\tau'}\vert_\zz.
\]
\end{enumerate}

The mean probability of the event described in (E1) is upper-bounded by
\[
  \EE\eckig{\sum_\zz W^n(\zz\vert X_{ij}^{\tau_1},Y_{ik}^{\tau_2})1_{\{(U_i,X_{ij}^{\tau_1},Y_{ik}^{\tau_2})\notin T_{p_W,\delta}^n\rvert_\zz\}}}.
\]
Note that the joint probability of the triple $(U_i,X_{ij}^{\tau_1},Y_{ik}^{\tau_2})$ and the channel output is $p_W^n$. Lemma \ref{Pinsker} from the Appendix then implies that the above term can be bounded by 
\begin{equation}\label{ersteabsch}
  (n+1)^{\lvert \U\rvert\lvert \X\rvert\lvert \Y\rvert\lvert \Z\rvert}2^{-nc\delta^2}.
\end{equation}

We now bound the probability that one of the events (E\ref{i})-(E\ref{k}) holds for some fixed $(\tau_1',\tau_2')$. To this end we use Lemma \ref{hitlemma}. The pair $(\U^n\times \X^n\times \Y^n,\mathcal{E})$, where $\mathcal{E}=\{E^{\tau'}\rvert_\zz:\zz\in \Z^n\}$,  defines a cubic hypergraph. Further, the collection of random vectors
\[
  \{(U_{i'},X_{i'j'}^{\tau_1'},Y_{i'k'}^{\tau_2'}):i',j',k'\}
\]
is a random $(M_0,M_1,M_2)$-half lattice on $\U^n\times \X^n\times \Y^n$. One obtains a probability measure on $\U^n\times\X^n\times \Y^n\times\mathcal{E}$ via
\[
  Q(\uu,\xx,\yy,E\rvert_\zz)=W^n(\zz\vert \xx,\yy)\PP[(U_i,X_{ij}^{\tau_1},Y_{ik}^{\tau_2})=(\uu,\xx,\yy)].
\]
We then obtain for fixed $(\tau_1',\tau_2')$ that 
\begin{align*}
  &\EE\left[\sum_\zz W^n(\zz\vert X_{ij}^{\tau_1},Y_{ik}^{\tau_2})1_{\{\textnormal{(E\ref{i}), (E\ref{jk}), (E\ref{j}), or (E\ref{k}) holds for }\tau'\}}\right]\\
  &\leq\sum_{\uu,\xx,\yy,\zz}Q(\uu,\xx,\yy,E\vert_\zz)\bigl(\PP[\text{(E\ref{i}) holds for }\tau'\vert U_i=\uu,X_{ij}=\xx,Y_{ik}=\yy]\\
  &\qquad\qquad\qquad\qquad+\PP[\text{(E\ref{jk}) holds for }\tau'\vert U_i=\uu,X_{ij}=\xx,Y_{ik}=\yy]\\
  &\qquad\qquad\qquad\qquad+\PP[\text{(E\ref{j}) holds for }\tau'\vert U_i=\uu,X_{ij}=\xx,Y_{ik}=\yy]\\
  &\qquad\qquad\qquad\qquad+\PP[\text{(E\ref{k}) holds for }\tau'\vert U_i=\uu,X_{ij}=\xx,Y_{ik}=\yy]\bigr).
\end{align*}
By the half-lattice property and Lemma \ref{hitlemma}, the above term can be upper-bounded by
\begin{align}
&M_0M_1M_2\max_\zz\PP[(U_1,X_{11}^{\tau_1'},Y_{11}^{\tau_2'})\in E^{\tau'}\rvert_\zz]\label{m0m1m2}\\
  &\quad+M_1M_2\max_{\zz,\uu}\PP[(X_{11}^{\tau_1'},Y_{11}^{\tau_2'})\in E^{\tau'}\rvert_{(\zz,\uu)}\vert U_1=\uu]\label{m1m2}\\
  &\quad+M_1\max_{\zz,\uu,\yy}\PP[X_{11}^{\tau_1'}\in E^{\tau'}\rvert_{(\zz,\uu,\yy)}\vert U_1=\uu]\label{m1}\\
  &\quad+M_2\max_{\zz,\uu,\xx}\PP[Y_{11}^{\tau_2'}\in E^{\tau'}\rvert_{(\zz,\uu,\xx)}\vert U_1=\uu].\label{m2}
\end{align}
It remains to bound the expressions \eqref{m0m1m2}-\eqref{m2}. For every $W\in \W_{\tau'}$, let the random vector $(U,X_{\tau_1'},Y_{\tau_2'},Z_W)$ have distribution $p_W$. We use Lemma \ref{Igor} a) and \ref{unionlemma} from the Appendix to bound \eqref{m0m1m2} by
\[
  M_0M_1M_2\,2^{-n(\inf_{W'\in \W_{\tau'}} I(Z_W; U,X_{\tau_1'},Y_{\tau_2'})-\zeta_1)}.
\]
This equals 
\begin{equation}\label{m0m1m2a}
  M_0M_1M_2\,2^{-n(\inf_{W'\in \W_{\tau'}} I(Z_W; X_{\tau_1'},Y_{\tau_2'})-\zeta_1)}
\end{equation}
because the sequence $(U,[X_{\tau_1},Y_{\tau_2}],Z_W)$ forms a Markov chain. Here, $\zeta_1$ is an error term which depends on $\delta$ and which converges to zero as $\delta$ tends to zero. Using Lemmas \ref{Igor} b) and Lemma \ref{unionlemma} from the Appendix, we see that the terms in \eqref{m1m2}-\eqref{m2} can be bounded by 
\begin{align}
  M_1M_2&\,2^{-n(\inf_{W'\in \W_{\tau'}}I(Z_W;X_{\tau_1'},Y_{\tau_2'}\vert U)-\zeta_2)},\label{m1m2a}\\
     M_1&\,2^{-n(\inf_{W'\in \W_{\tau'}}I(Z_W,Y_{\tau_2'};X_{\tau_1'}\vert U)-\zeta_3)},\label{m1a}\\
     M_2&\,2^{-n(\inf_{W'\in \W_{\tau'}}I(Z_W,X_{\tau_1'};Y_{\tau_2'}\vert U)-\zeta_4)},\label{m2a}
\end{align}
respectively. Here, again, $\zeta_2,\zeta_3,\zeta_4$ depend on $\delta$ and converge to zero as $\delta$ tends to zero. The bounds in \eqref{m1a} and \eqref{m2a} can be reduced to 
\begin{align}
  M_1&\,2^{-n(\inf_{W'\in \W_{\tau'}}I(Z_W;X_{\tau_1'}\vert Y_{\tau_2'},U)-\zeta_3)},\label{m1aa}\\
  M_2&\,2^{-n(\inf_{W'\in \W_{\tau'}}I(Z_W;Y_{\tau_2'}\vert X_{\tau_1'},U)-\zeta_4)}.\label{m2aa}
\end{align}
For \eqref{m1aa}, this follows from
\[
  I(Z_W,Y_{\tau_2'};X_{\tau_1'}\vert U)=I(Y_{\tau_2'};X_{\tau_1'}\vert U)+I(Z_W;X_{\tau_1'}\vert Y_{\tau_2'},U)=I(Z_W;X_{\tau_1'}\vert Y_{\tau_2'},U),
\]
where the chain rule for mutual information was used and the fact that $X_{\tau_1'}$ and $Y_{\tau_2'}$ are conditionally independent given $U$. The bound \eqref{m2aa} follows in an analogous way. Collecting \eqref{ersteabsch} and, for each $(\tau_1',\tau_2')\in T_1\times T_2$, the bounds \eqref{m0m1m2a}, \eqref{m1m2a}, \eqref{m1aa}, and \eqref{m2aa}, we obtain an upper bound for the mean maximal error defined in \eqref{meanmax} of
\begin{align*}
  &(n+1)^{\lvert \X\rvert\lvert \Y\rvert\lvert \Z\rvert\lvert \U\rvert}2^{-nc\delta^2}\\
  &+\lvert T_1\rvert\lvert T_2\rvert M_0M_1M_2\,2^{-n(\min_{\tau'\in T_1\times T_2}\inf_{W'\in \W_{\tau'}} I(Z_W; X_{\tau_1'},Y_{\tau_2'})-\zeta_1)}\\
  &+\lvert T_1\rvert\lvert T_2\rvert M_1M_2\,2^{-n(\min_{\tau'\in T_1\times T_2}\inf_{W'\in \W_{\tau'}}I(Z_W;X_{\tau_1'},Y_{\tau_2'}\vert U)-\zeta_2)}\\
  &+\lvert T_1\rvert\lvert T_2\rvert M_1\,2^{-n(\min_{\tau'\in T_1\times T_2}\inf_{W'\in \W_{\tau'}}I(Z_W;X_{\tau_1'}\vert Y_{\tau_2'},U)-\zeta_3)}\\
  &+\lvert T_1\rvert\lvert T_2\rvert M_2\,2^{-n(\min_{\tau'\in T_1\times T_2}\inf_{W'\in \W_{\tau'}}I(Z_W;Y_{\tau_2'}\vert X_{\tau_1'},U)-\zeta_4)}.
\end{align*}
Note that this bound is uniform in $W$. It tends to zero exponentially with rate $\tilde\zeta>0$ if
\begin{equation}\label{ratenbed}
\begin{aligned}
  \frac{1}{n}\log(M_0M_1M_2)&<\min_{\tau'\in T_1\times T_2}\inf_{W'\in \W_{\tau'}}I(Z_W; X^{\tau_1'},Y^{\tau_2'})-\zeta_1-\tilde\zeta,\\
  \frac{1}{n}\log(M_1M_2)&<\min_{\tau'\in T_1\times T_2}\inf_{W'\in \W_{\tau'}}I(Z_W;X^{\tau_1'},Y^{\tau_2'}\vert U)-\zeta_2-\tilde\zeta,\\
  \frac{1}{n}\log M_1&<\min_{\tau'\in T_1\times T_2}\inf_{W'\in \W_{\tau'}}I(Z_W;X^{\tau_1'}\vert Y^{\tau_2'},U)-\zeta_3-\tilde\zeta,\\
  \frac{1}{n}\log M_2&<\min_{\tau'\in T_1\times T_2}\inf_{W'\in \W_{\tau'}}I(Z_W;Y^{\tau_2'}\vert X^{\tau_1'},U)-\zeta_4-\tilde\zeta,
\end{aligned}
\end{equation}
for some $\delta>0$.

Now assume that $(R_0,R_1,R_2)$ is contained in $\mathcal{C}_\CM^*(\W,t_1,t_2)$. Hence, there is a $p\in\Pi_1$ such that 
\[
  (R_0,R_1,R_2)\in\bigcap_{(\tau_1',\tau_2')}\;\bigcap_{W'\in \W_{\tau_1'\tau_2'}}\mathcal{R}_\CM(p,\tau_1',\tau_2',W').
\]
For $n$ large, we can find numbers $M_0,M_1,M_2$ satisfying 
\[
  R_\nu-\eps\leq\frac{1}{n}\log M_\nu\leq R_\nu-\frac{\eps}{2}.
\]
Choose $\delta$ and $\tilde\zeta$ such that $\zeta_1\wedge\zeta_2\wedge\zeta_3\wedge\zeta_4+\tilde\zeta\leq\eps/2$. Inserting this in \eqref{ratenbed} establishes the existence of a sequence of random codes whose mean average error converges to 0 with rate $\tilde\zeta$. Hence, for every $(R_0,R_1,R_2)\in\mathcal{C}^*(\W,t_1,t_2)$, one can find random codes according to the procedure described above with rates close to $(R_0,R_1,R_2)$ and with an exponentially small maximum error probability. 

\subsubsection{Extracting a Deterministic Code for $\lvert \W\rvert<\infty$}\label{hendl}

The next step is to extract a deterministic code with the same rate triple and with small average error from the random one. This is easy when $\lvert \W\rvert<\infty$, an approximation argument similar to the one in \cite{BBT} solves the problem for $\lvert \W\rvert=\infty$. So let us first assume that $\lvert \W\rvert<\infty$. For $\tau=(\tau_1,\tau_2)\in T_1\times T_2$ and $W\in \W_\tau$, we define on the underlying probability space $(\Omega,\mathcal{F},\PP)$ the random variable
\[
  P_e^W(\omega):=\frac{1}{M_0M_1M_2}\sum_{i,j,k}W^n(F_{ijk}^c(\omega)\vert X_{ij}^{\tau_1}(\omega),Y_{ik}^{\tau_2}(\omega))
\]
This gives the average error for a channel state $W\in \W_\tau$ and the random code determined by the elementary event $\omega\in\Omega$. For every $W\in \W$ and every $(R_0,R_1,R_2)$ in $\mathcal{C}_\CM^*(\W,t_1,t_2)$, we found above a random code with block length $n$ and message set $[1,M_0^{(n)}]\times[1,M_1^{(n)}]\times[1,M_2^{(n)}]$, and a $\tilde\zeta>0$ such that 
\[
  \mathbb{E}[P_e^W]\leq2^{-n\tilde\zeta}
\]
and $(1/n)\log M_\nu^{(n)}\geq R_\nu-\eps$ for $\nu=0,1,2$ if $n$ is large (the bound on the mean maximum error \textit{a fortiori} also holds for the mean average error). For $0<\zeta<\tilde\zeta$, define the set 
\[
  \Omega_W:=\{\omega\in\Omega:P_e^W(\omega)\leq 2^{-n\zeta}\}.
\]
If $\bigcap_{W\in \W}\Omega_W$ is nonempty, we can infer the existence of a deterministic code$_\CM(n,M_0^{(n)},M_1^{(n)},M_2^{(n)},2^{-n\zeta})$ with exponentially small error probability. And indeed, the Markov inequality implies
\begin{align*}
  \mathbb{P}\Bigl[\bigcap_{W\in \W}\Omega_W\Bigr]&\geq1-\sum_{W\in \W}\mathbb{P}[\Omega_W^c]\\
     &\geq1-2^{n\zeta}\sum_{W\in \W}\mathbb{E}[P_e^W]\\
     &\geq1-\lvert \W\rvert2^{-n(\tilde\zeta-\zeta)}>0,
\end{align*}
so $\bigcap_{W}C_W$ must be nonempty. This proves the existence of a deterministic code$_\CM(n,M_0^{(n)},M_1^{(n)},M_2^{(n)},2^{-n\zeta})$ with exponentially decaying average error probability for every $(R_0,R_1,R_2)\in\mathcal{C}^*(\W,t_1,t_2)$, so this whole set is achievable.

\subsubsection{Approximation for $\lvert \W\rvert=\infty$}

For a positive integer $N$ to be chosen later, we first define an approximating compound discrete memoryless MAC. For every $\tilde W\in \W_N$, $\tilde W(z\vert x,y)$ is a multiple of $(2N\lvert T_1\rvert\lvert T_2\rvert)^{-1}$ for all $x\in \X,y\in \Y,z\in \Z$. Clearly, $\lvert \W_N\rvert\leq(2N\lvert T_1\rvert\lvert T_2\rvert+1)^{\lvert \X\rvert\lvert \Y\rvert\lvert \Z\rvert}$. The following is a slight variation of \cite[Lemma 4]{BBT}.

\begin{lem}\label{BBT}
  For every $N>2\lvert \Z\rvert$, there is a function $f:\W\rightarrow \W_N$ satisfying $f(\W_{\tau})\cap f(\W_{\tau'})=\varnothing$ if $\tau\neq\tau'$ such that for every $W\in \W$,
\begin{align}
  &\lvert W(z\vert x,y)-f(W)(z\vert x,y)\rvert\leq\dfrac{\lvert \Z\rvert}{N},\label{BBT1}\\
  &W(z\vert x,y)\leq\exp\left(\dfrac{2\lvert \Z\rvert^2}{N}\right)f(W)(z\vert x,y).\label{BBT2}
\end{align}
\end{lem}
Let $N$ be as in the lemma and let $f_N$ be the corresponding function from $\W$ to $\W_N$. Let $p\in\Pi_1$, $\tau=(\tau_1,\tau_2)\in T_1\times T_2$, and $W\in \W_\tau$. By \eqref{BBT1} and \cite[Lemma 1.2.7]{CK} (which quantifies the uniform continuity of entropy), one has the inequalities
\begin{align*}
  \lvert I(Z_W;X_{\tau_1},Y_{\tau_2})-I(Z_{f_N(W)};X_{\tau_1},Y_{\tau_2})\rvert&\leq-2\frac{\lvert \Z\rvert^3}{N}\log\frac{\lvert \Z\rvert^2}{N},\\
  \lvert I(Z_W;X_{\tau_1},Y_{\tau_2}\vert U)-I(Z_{f_N(W)};X_{\tau_1},Y_{\tau_2}\rvert U)\rvert&\leq-2\frac{\lvert \Z\rvert^3}{N}\log\frac{\lvert \Z\rvert^2}{N},\\
  \lvert I(Z_W;X_{\tau_1}\vert Y_{\tau_2},U)-I(Z_{f_N(W)};X_{\tau_1}\vert Y_{\tau_2},U)\rvert&\leq-2\frac{\lvert \Z\rvert^3}{N}\log\frac{\lvert \Z\rvert^2}{N},\\
  \lvert I(Z_W;Y_{\tau_2}\vert X_{\tau_1},U)-I(Z_{f_N(W)};Y_{\tau_2}\vert X_{\tau_1},U)\rvert&\leq-2\frac{\lvert \Z\rvert^3}{N}\log\frac{\lvert \Z\rvert^2}{N}.
\end{align*}
Now fix a triple $(R_0,R_1,R_2)$ which is contained in the interior of $\mathcal{C}_\CM^*(\W,t_1,t_2)$. The above inequalities imply that for large $N$ it is contained in the interior of $\mathcal{C}_\CM^*(f_N(\W),\tilde t_1,\tilde t_2)$ defined through the channel $(f_N(\W),\tilde t_1,\tilde t_2)$. Here, the necessarily finite partitions $\tilde t_\nu=\{\tilde\W_{\tau_\nu}\subset\W_N:\tau_\nu\in T_\nu\}$ $(\nu=1,2)$ of $\W_N$ are defined by
\[
  \tilde\W_{\tau_\nu}=f_N(\W_{\tau_1}).
\]
recall \eqref{vbar}. These really are partitions by Lemma \ref{BBT}. The achievability result in \ref{hendl} established the existence of codes$_\CM(n,M_0^{(n)},M_1^{(n)},M_2^{(n)},2^{-n\zeta})$ for the compound MAC $(f_N(\W),\tilde t_1,\tilde t_2)$ such that
\[
  \frac{1}{n}\log M_\nu^{(n)}\geq R_\nu+\frac{2\lvert \Z\rvert^3}{3N}\log\frac{\lvert \Z\rvert^2}{N}-\frac{\eps}{2}.
\]
For $N$ large enough, one has $(1/n)\log M_\nu^{(n)}\geq R_\nu-\eps$. Then, the above sequence of codes for $(\W_N,\tilde t_1,\tilde t_2)$ has the desired rates for $(\W,t_1,t_2)$. It remains to bound the average error incurred when applying the codes for transmission over $(\W,t_1,t_2)$. For fixed $n$, let the code$_\CM(n,M_0^{(n)},M_1^{(n)},M_2^{(n)},2^{-n\zeta})$ have the form \eqref{codecmform}. For any $W\in\W_{\tau_1\tau_2}$, \eqref{BBT2} implies that the average error can be bounded by 
\begin{align*}
  \frac{1}{M_0M_1M_2}\sum_{i,j,k}W^n(F_{ijk}^c\vert\x{ij}{\tau_1},\y{ik}{\tau_2})&\leq e^{2n\lvert \Z\rvert^2/N}\frac{1}{M_0M_1M_2}\sum_{i,j,k}f_N(W)^n(F_{ijk}^c\vert\x{ij}{\tau_1},\y{ik}{\tau_2})\\
  &\leq\exp\rund{-n\rund{\zeta\ln 2-\frac{2\lvert \Z\rvert}{N}}}.
\end{align*}
By enlarging $N$ if necessary, this goes to zero as $n$ approaches infinity, so one obtains an exponentially small average probability of error. One checks easily that the existence of a sequence of codes$_\CM(n,M_0^{(n)},M_1^{(n)},M_2^{(n)},2^{-n\zeta})$ with $(1/n)\log M_\nu^{(n)}\geq R_\nu-\eps$ for every $(R_0,R_1,R_2)$ in the interior of $\mathcal{C}_\CM^*(\W,t_1,t_2)$ implies the existence of such a sequence also for the rate triples lying on the boundary of $\mathcal{C}_\CM^*(\W,t_1,t_2)$.

\subsubsection{Convexity and Bound on $\lvert \U\rvert$}

The convexity of $\mathcal{C}^*(\W,t_1,t_2)$ is clear by the concavity of mutual information in the input distributions. The bounds on $\lvert \U\rvert$ follow in the same way as in \cite{Wi1}.

\subsection{The MAC with Conferencing Encoders}

The achievability part of Theorem \ref{thmconf} relies on the achievability part of Theorem \ref{thmcm}. We first define the Willems conferencing functions that will turn out to be optimal for large blocklengths in the course of the proof. Then, we show how Theorem \ref{thmcm} can be applied to design a code$_\CONF$ from a code$_\CM$ using these conferencing functions if certain conditions on the rates are fulfilled. Next, we show that these conditions can be fulfilled. Finally, we show that the average error of the conferencing codes thus defined is small. As in the achievability proof for the MAC with common message, it suffices to assume that the receiver has no CSIR.

\subsubsection{Preliminary Considerations}\label{prelcons}

Let $[1,M_1]$ and $[1,M_2]$ be message sets, let $n$ be a blocklength, and let $C_1,C_2$ be conferencing capacities. If $n$ is large enough, we can construct a pair of simple one-shot Willems conferencing functions (cf. Example \ref{iter}) with these message sets which will be admissible with respect to $n$ and $C_1,C_2$. The blocklength needs to be large enough to ensure the existence of positive integers $V_1,V_2$ with
\begin{equation}\label{optconfcond}
  \frac{1}{n}\log\lvert T_\nu\rvert\leq\frac{1}{n}\log V_\nu\leq C_\nu\qquad(\nu=1,2).
\end{equation}
Then define
\begin{align*}
  \mu_\nu&:=\left\lfloor\frac{V_\nu}{\lvert T_\nu\rvert}\right\rfloor\wedge M_\nu
\intertext{and}
  \xi_\nu&:=\begin{cases}
              \left\lfloor\frac{M_\nu-1}{\mu_\nu-1}\right\rfloor&\quad\text{if }\mu_\nu\geq 2\\
              0                                                 &\quad\text{if }\mu_\nu=1.
            \end{cases}
\end{align*}
Every $\ell_\nu\in[1,M_\nu]$ can be written uniquely as
\begin{equation}\label{messrep}
  \ell_\nu=(i_\nu-1)\xi_\nu+\ell_\nu',
\end{equation}
where $i_\nu\in[1,\mu_\nu]$ and where
\begin{align*}
  \ell_\nu'&\in\begin{cases}
         \eckig{1,\xi_\nu}&\text{ if }i_\nu\leq \mu_\nu-1,\\
         \eckig{1,M_\nu-(\mu_\nu-1)\xi_\nu}&\text{ if }i_\nu=\mu_\nu.
       \end{cases}
\end{align*}
The conferencing function $g_\nu:[1,M_\nu]\times T_\nu\rightarrow [1,\mu_\nu]\times T_\nu$ can now be defined by
\begin{equation}\label{defg}
  g_\nu(\ell_\nu,\tau_\nu)=(i_\nu,\tau_\nu)\quad\text{if }\ell_\nu=(i_\nu-1)\xi_\nu+\ell_\nu'.
\end{equation}
Note that by \eqref{optconfcond},
\begin{equation}\label{confcard}
  \frac{1}{n}\log\lvert [1,\mu_\nu]\times T_\nu\rvert\leq\frac{1}{n}\log V_\nu\leq C_\nu,
\end{equation}
so $g_\nu$ is an admissible one-shot Willems conferencing function.

\subsubsection{Coding for $C_1,C_2>0$}\label{subsect:coding}

Now we show how to construct a code$_\CONF$ using the conferencing functions defined above and the codes$_\CM$ whose existence was proved in \ref{subsect:CMach}. We assume $C_1,C_2>0$. Let $(R_1,R_2)$ be contained in $\mathcal{C}_\CONF^*(\W,t_1,t_2,C_1,C_2)$. Set
\begin{align*}
  \tilde R_\nu&:=R_\nu\wedge C_\nu,\quad&\nu=1,2,\\
  R_\nu'&:=R_\nu-\tilde R_\nu,\quad&\nu=1,2,\\
  R_0'&:=\tilde R_1+\tilde R_2.
\end{align*}
Then $(R_0',R_1',R_2')$ is contained in $\mathcal{C}_\CM^*(\W,t,t)$, defined through $(\W,t,t)$, where the CSIT partition $t$ of both encoders is given by
\begin{equation}\label{gemkan}
  t=\{W_{\tau_1\tau_2}:(\tau_1,\tau_2)\in T_1\times T_2\}.
\end{equation}
One knows by Theorem \ref{thmcm} that for any $\eps>0$, there is a $\zeta=\zeta(\eps)$ such that for large $n$, there is a code$_\CM$ $(n,M_0^{(n)},M_1^{(n)},M_2^{(n)},2^{-n\zeta})$ for $(\W,t,t,r)$ with
\begin{equation}\label{rates}
  R_\nu'\geq\frac{1}{n}\log M_\nu^{(n)}\geq R_\nu'-\eps.
\end{equation}
For fixed $n$, let such a code$_\CM$ have the form
\begin{equation}\label{cmcodeform}
  \{(\tx{ij}{\tau},\ty{ik}{\tau},\tilde F_{ijk}):(i,j,k)\in[1,M_0^{(n)}]\times[1,M_1^{(n)}]\times[1,M_2^{(n)}],\tau\in T_1\times T_2\}.
\end{equation}

In \ref{cases}, we will show that if $n$ is large enough, one can find $M_1,M_2$ and $V_1,V_2$ such that \eqref{optconfcond} is satisfied for $\nu=1,2$ and such that
\begin{align}
  \frac{1}{n}\log \mu_1&\leq \tilde R_1,\label{C_1}\\
  \frac{1}{n}\log \mu_2&\leq \tilde R_2\label{C_2}
\end{align}
and
\begin{align}
  \frac{M_0^{(n)}}{2\lvert T_1\rvert\lvert T_2\rvert}\leq\mu_1\mu_2&\leq M_0^{(n)},\label{R_0}\\
  \xi_1&=M_1^{(n)},\label{R_1}\\
  \xi_2&=M_2^{(n)}\label{R_2}.
\end{align}
Because of the validity of \eqref{optconfcond}, one can carry out the construction of the conferencing functions described in \ref{prelcons}. As noted in \eqref{confcard}, the pair $(g_1,g_2)$ defined in \eqref{defg} for $\nu=1,2$ is an admissible pair of conferencing functions. By \eqref{R_0}-\eqref{R_2}, one can naturally consider the set $[1,\mu_1]\times[1,\mu_2]$ as a subset of $[1,M_0^{(n)}]$ and the sets $[1,\xi_\nu]$ as equal to $[1,M_\nu^{(n)}]$. With $(g_1,g_2)$ and recalling the alternative definition of codes$_\CONF$ right after Definition \ref{codeconf}, one can now define a code$_\CONF$ by a family as in \eqref{codeseq} as follows: assume that $j\in[1,M_1]$ and $k\in[1,M_2]$ have a representation $(i_1,j')$ and $(i_2,k')$ as in \eqref{messrep}. Then
\begin{align}
  \x{jk}{\tau_1\tau_2}&:=\tx{(i_1,i_2)j'}{\tau_1\tau_2},\label{xes}\\
  \y{jk}{\tau_1\tau_2}&:=\ty{(i_1,i_2)k'}{\tau_1\tau_2},\label{yse}
\end{align}
where $(i_1,i_2)$ is to be considered an element of $[1,M_0^{(n)}]$. The decoding sets are defined as
\begin{equation}\label{F_jk}
  F_{jk}:=\tilde F_{(i_1,i_2)j'k'}.
\end{equation}
This code is a code$_\CONF(n,M_1,M_2,C_1,C_2)$ for the compound MAC with conferencing encoders as in Definition \ref{codeconf} because it satisfies \eqref{xe} and \eqref{ys} for the pair of conferencing functions $(g_1,g_2)$. We now show that it also achieves the desired rates. Without loss of generality, one can assume that 
\begin{equation}\label{assuzeta}
  \frac{1}{n}\log(\mu_\nu-1)\geq\frac{1}{n}\log\mu_\nu-\frac{\zeta}{4}\wedge\frac{\eps}{2}
\end{equation}
if $\mu_\nu>1$. We may also assume that 
\begin{equation}\label{also}
  \frac{1}{n}\log(2\lvert T_1\rvert\lvert T_2\rvert)\leq\eps.
\end{equation}
It follows for large enough $n$ from \eqref{rates} and \eqref{R_0}-\eqref{also} and the definition of the $R_\nu'$ that
\begin{equation}\label{M1M2}
\begin{aligned}
  \frac{1}{n}\log M_1M_2  &\geq\frac{1}{n}\log(M_0^{(n)}M_1^{(n)}M_2^{(n)})-\frac{1}{n}\log(2\lvert T_1\rvert\lvert T_2\rvert)-\frac{\zeta}{2}\wedge\eps\\
                        &\geq R_1+R_2-5\eps.
\end{aligned}
\end{equation}
Further by \eqref{C_1}, \eqref{C_2}, \eqref{R_1}, \eqref{R_2}, and \eqref{rates}, for $\nu=1,2$,
\begin{align}\label{M1}
  \frac{1}{n}\log M_\nu\leq\frac{1}{n}\log\mu_\nu+\frac{1}{n}\log\xi_\nu\leq \tilde R_\nu+R_\nu'=R_\nu.
\end{align}
Combining \eqref{M1M2} and \eqref{M1} yields
\begin{align*}
  \frac{1}{n}\log M_\nu&\geq R_\nu-5\eps,
\end{align*}
so the rates are as desired. In Subsection \ref{error}, the average error of this code$_\CONF(n,M_1,M_2,C_1,C_2)$ will be shown to be small, thus finishing the proof of the achievability part of Theorem \ref{thmconf}.

\subsubsection{Finding $M_1,M_2,V_1,V_2$ for $C_1,C_2>0$}\label{cases}

Let a positive integer $n$ be fixed. Without loss of generality, let $0<\eps<\tilde R_1\wedge\tilde R_2$, so again without loss of generality, one can assume $\tilde R_1<(1/n)\log M_0^{(n)}$. We choose
\[
  V_1=\lvert T_1\rvert\left\lfloor\frac{2^{n\tilde R_1}}{\lvert T_1\rvert}\right\rfloor
  \qquad\text{and}\qquad
  V_2=\lvert T_2\rvert\left\lfloor\frac{2^{-n\tilde R_1}M_0^{(n)}}{\lvert T_2\rvert}\right\rfloor.
\]
Hence \eqref{optconfcond} and \eqref{C_1}-\eqref{R_0} are always satisfied. In order to find $M_1$ and $M_2$, three cases need to be distinguished. In all of the cases, it is straightforward to check that \eqref{R_1} and \eqref{R_2} hold.

\textit{Case 1: $\tilde R_\nu=R_\nu$ for $\nu=1,2$.} Then $R_1'=R_2'=0$. Set $M_\nu=V_\nu/\lvert T_\nu\rvert$ for $\nu=1,2$. Then $\mu_\nu=M_\nu$, so $\xi_1=\xi_2=1$.

\textit{Case 2: $\tilde R_\nu=C_\nu$ for $\nu=1,2$.} Choose $M_\nu$ such that
\[
  \left\lfloor\frac{M_\nu-1}{V_\nu/\lvert T_\nu\rvert-1}\right\rfloor=M^{(n)}_\nu
\]
for $\nu=1,2$. Then $\mu_\nu=V_\nu/\lvert T_\nu\rvert$ and $\xi_\nu=M_\nu^{(n)}$.

\textit{Case 3a: $\tilde R_1=C_1$, $\tilde R_2=R_2$.} Then $R_2'=0$ and $R_2\leq C_2$. Choose $M_2=V_2/\lvert T_2\rvert$ and $M_1$ such that
\[
  \xi_1=\left\lfloor\frac{M_1-1}{V_1/\lvert T_1\rvert-1}\right\rfloor=M^{(n)}_1.
\]
Then $\mu_\nu=V_\nu/\lvert T_\nu\rvert$ for both $\nu$, and note that $\xi_1=M_1^{(n)}$ and $\xi_2=1$.

\textit{Case 3b: $\tilde R_2=C_2$, $\tilde R_1=R_1$.} Analogous to case 3.

\subsubsection{The average error for $C_1,C_2>0$}\label{error}

Recall the form \eqref{cmcodeform} of the code$_\CM(n,M_0^{(n)},M_1^{(n)},M_2^{(n)},2^{-n\zeta})$ and the definitions \eqref{xes}-\eqref{F_jk} of the code$_\CONF(n,M_1,M_2,C_1,C_2)$ in \ref{subsect:coding}. We now bound its average error. Let the channel state $W\in \W$ be arbitrary. The code$_\CM$ satisfies
\[
  \frac{1}{M_0^{(n)}M_1^{(n)}M_2^{(n)}}\sum_{i,j',k'}W^n(\tilde F_{ij'k'}^c\vert\tx{ij'}{\tau},\ty{ik'}{\tau})\leq2^{-n\zeta},
\]
where the sum ranges over the message set $[1,M_0^{(n)}]\times[1,M_1^{(n)}]\times[1,M_2^{(n)}]$ of the code$_\CM$. With assumption \eqref{assuzeta} and \eqref{R_0}-\eqref{R_2}, one has
\begin{align*}
  M_1M_2&\geq2^{-n(\zeta/2)+1}\lvert T_1\rvert\lvert T_2\rvert M_0^{(n)}M_1^{(n)}M_2^{(n)}.
\end{align*}
One thus obtains for the average error of the code$_\CONF(n,M_1,M_2,C_1,C_2)$ that
\begin{align*}
  &\mathrel{\hphantom{\leq}}\frac{1}{M_1M_2}\sum_{j,k\in[1,M_1]\times[1,M_2]}W^n(F_{jk}^c\vert\x{jk}{\tau},\y{jk}{\tau})\\
  &\leq\frac{2^{n(\zeta/2)}}{2\lvert T_1\rvert\lvert T_2\rvert}\cdot\frac{1}{M_0^{(n)}M_1^{(n)}M_2^{(n)}}\sum_{(i,j',k')\in[1,M_0^{(n)}]\times[1,M_1^{(n)}]\times[1,M_2^{(n)}]}W^n(\tilde F_{ij'k'}^c\vert\tx{ij'}{\tau},\ty{ik'}{\tau})\\
  &\leq\frac{2^{-n(\zeta-\zeta/2)}}{2\lvert T_1\rvert\lvert T_2\rvert}.
\end{align*}
This proves that the average error of this code$_\CONF(n,M_1,M_2,C_1,C_2)$ is exponentially small. Thus the rate pair $(R_1,R_2)$ is achievable, and this finishes the proof of the achievability part of Theorem \ref{thmconf} for the case $C_1,C_2>0$.

\subsubsection{The case $C_1>0,C_2=0$}

First note that the case $C_1=0,C_2>0$ is analogous to the case $C_1>0,C_2=0$ which is treated here. One can use all the methods used in the case $C_1,C_2>0$ for the first user. An admissible one-shot Willems conferencing function $g_1$ can be constructed as in \ref{prelcons}. Then let $(R_1,R_2)\in\mathcal{C}_{\CONF,1}^*(\W,t_1,t_2,C_1)$. One checks that the triple $(R_0',R_1',R_2')$ defined as in \ref{subsect:coding} is contained in $\mathcal{C}_\CM(\W,t_1,t)$, where $t$ also is defined as in \ref{subsect:coding}. Given a blocklength $n$, one then can find $M_1,M_2,V_1,V_2$ as in \ref{cases}, where only the relevant cases need to be considered. This then defines a good code$_\CONF$.

\subsubsection{Convexity and Bound on $\lvert \U\rvert$}

The convexity of $\mathcal{C}_\CONF(\W,t_1,t_2,r,C_1,C_2)$ is inherited from the convexity of $\mathcal{C}_\CM(\W,t_1,t_2,r)$. Also the bound on the cardinality of the set $\U$ appearing in the parametrization of the rate regions comes from the bound on the range of the auxiliary random variable appearing in the parametrization of the capacity region of the compound MAC with common message.


\section{The Converses}\label{sect:conv}

We will concentrate on the converse for the MAC with conferencing encoders because it requires some non-standard preliminaries. For the converse for the MAC with common message, we only show how to start the proof, the rest is similar to the proof of the MAC with conferencing encoders. For both outer bounds, one assumes perfect CSIR. As we will prove that, fixing a pair of CSIT partitions, this outer bound coincides with the inner bound with no CSIR, this includes all possible permissible types of CSIR.

\subsection{The Converse for the MAC with Conferencing Encoders}

First we define what we mean exactly by the statement that a weak converse holds for $(\W,t_1,t_2,r)$ with codes$_\CONF$.
\begin{defn}
A weak converse holds for the compound MAC $(\W,t_1,t_2,r)$ with codes$_\CONF$ if the average error $\lambda$ of every code$_\CONF(n,M_1,M_2,C_1,C_2,\lambda)$ whose rate pair $((1/n)\log M_1,(1/n)\log M_2)$ is further than $\eps>0$ from $\mathcal{C}_\CONF^*(\W,t_1,t_2,C_1,C_2)$ satisfies $\lambda\geq\lambda(\eps)>0$ if $n$ is large enough. Without loss of generality we measure distance in the $\ell^1$-norm, so the statement that the rate pair of the code is further than $\eps$ from $\mathcal{C}_\CONF^*(\W,t_1,t_2,C_1,C_2)$ can be formulated as 
\begin{equation}\label{ratenabstand}
  \min_{(R_1,R_2)\in\mathcal{C}_\CONF^*(\W,t_1,t_2,C_1,C_2)}\left\{\left\lvert\frac{1}{n}\log M_1-R_1\right\rvert+\left\lvert\frac{1}{n}\log M_2-R_2\right\rvert\right\}\geq\eps.
\end{equation}
\end{defn}
In \ref{subsubsect:auxsit}, we show that the weak converse for the compound MAC with conferencing encoders is implied by the weak converse for an auxiliary compound MAC with different CSIT and a slightly restricted kind of cooperation. CSIR will also be assumed to be perfect for that channel. In \ref{subsubsect:convaux}, we then show that the weak converse holds for this auxiliary MAC. Throughout the section, we will assume that $C_1,C_2>0$. The case of one conferencing capacity being equal to zero is treated analogously.

\subsubsection{An Auxiliary MAC}\label{subsubsect:auxsit}

We now describe the auxiliary MAC. Let $(\W,t_1,t_2,r)$ be given. As we assume perfect CSIR, we may assume $r=\{\{W\}:W\in\W\}$. Let $t_1,t_2$ be CSIT partitions as in \eqref{CSIT} and define the CSIT partition $t$ as in \eqref{gemkan}. Let the channel $(\W,t,t,r)$ be given (symmetric CSIT!). We now define what we mean by a 	code$_\AUX(n,M_1,M_2,\tilde C_1,\tilde C_2)$ for $(\W,t,t,r)$, where $n,M_1,M_2$ are positive integers and $\tilde C_1,\tilde C_2>0$. 

\begin{defn}
A code$_\AUX(n,M_1,M_2,\tilde C_1,\tilde C_2)$ is a quadruple $(\tilde f_1,\tilde f_2,\tilde g,\tilde\Phi)$ of functions which satisfy
\begin{align*}
  \tilde f_1&:[1,M_1]\times\tilde\Gamma\times T_1\times T_2\rightarrow\X^n,\\
  \tilde f_2&:[1,M_2]\times\tilde\Gamma\times T_1\times T_2\rightarrow\Y^n,\\
  \tilde g&:[1,M_1]\times[1,M_2]\rightarrow\tilde\Gamma,\\
  \tilde\Phi&:\Z^n\times R\rightarrow[1,M_1]\times[1,M_2],
\end{align*}
where $\tilde\Gamma$ is a finite set and where $\tilde g$ satisfies
\begin{align}
  \frac{1}{n}\log\lVert\tilde g_{j}\rVert\leq\tilde C_2\qquad\text{for all }j\in[1,M_1],\\
  \frac{1}{n}\log\lVert\tilde g_{k}\rVert\leq\tilde C_1\qquad\text{for all }k\in[1,M_1]
\end{align}
for the functions $\tilde g_{j}$ and $\tilde g_{k}$ defined by $\tilde g_{j}(j,k)=\tilde g_{k}(j,k)=\tilde g(j,k)$. The number $n$ is called the blocklength of the code.
\end{defn}

Thus an auxiliary code is one where only messages are exchanged, and where this is done independently of the CSIT. As codes$_\CONF$, every code$_\AUX$ can also be described by a family analogous to \eqref{codeseq} and a conferencing MAC like the $\tilde g$ from the above definition.

\begin{defn}
The code$_\AUX(n,M_1,M_2,\tilde C_1,\tilde C_2)$ is a code$_\AUX(n,M_1,M_2,\tilde C_1,\tilde C_2,\lambda)$ if
\[
  \max_{(\tau_1,\tau_2)\in T_1\times T_2}\;\sup_{W\in\W_{\tau_1\tau_2}}\frac{1}{M_1M_2}\sum_{j,k}W^n\bigl((\tilde F_{jk}^W)^c\vert\tx{jk}{\tau},\ty{jk}{\tau}\bigr)\leq\lambda.
\]
\end{defn}

In Subsubsection \ref{subsubsect:convaux} we will show a weak converse for $(\W,t,t,r)$ with codes$_\AUX$:
\begin{lem}\label{weakconvaux}
  Let a code$_\AUX(n,M_1,M_2,C_1+\delta,C_2+\delta,\lambda)$ be given with
\begin{equation}\label{l1aux}
  \min_{(R_1,R_2)\in\mathcal{C}_\CONF^*(\W,t,t,C_1+\delta,C_2+\delta)}\left\{\left\lvert\frac{1}{n}\log M_1-R_1\right\rvert+\left\lvert\frac{1}{n}\log M_1-R_1\right\rvert\right\}\geq\eps
\end{equation}
for some $\eps>0$. Then there is a $\lambda(\eps,\delta)>0$ such that $\lambda\geq\lambda(\eps,\delta)$ for sufficiently large $n$.
\end{lem}

We will show this lemma in \ref{subsubsect:convaux}. This together with the next lemma shows a weak converse as claimed in Theorem \ref{thmconf}.

\begin{lem}\label{diffsit}
For every $\delta>0$ there exists a positive integer $n_0$ such that for every $n\geq n_0$ and every code$_\CONF$ $(n,M_1,M_2,C_1,C_2,\lambda)$ for $(\W,t_1,t_2,r)$ there is a code$_\AUX(n,M_1,M_2,C_1+\delta,C_2+\delta,\lambda)$ for $(\W,t,t,r)$.
\end{lem}

\begin{proof}[Deduction of the weak converse for Theorem \ref{thmconf} from the weak converse for the auxiliary MAC]
  Before proving Lemma \ref{diffsit}, we show how it implies a weak converse for $(\W,t_1,t_2,r)$ with codes$_\CONF$. Assume that the code$_\CONF(n,M_1,M_2,C_1,C_2,\lambda)$ for $(\W,t_1,t_2,r)$ satisfies \eqref{ratenabstand}. Let $\delta>0$ be arbitrary. By Lemma \ref{diffsit}, for this code$_\CONF$, there is a code$_\AUX(n,M_1,M_2,C_1+\delta,C_2+\delta,\lambda)$ for $(\W,t,t,r)$. As $\mathcal{C}_\CONF^*(\W,t_1,t_2,\tilde C_1,\tilde C_2)=\mathcal{C}_\CONF^*(\W,t,t,\tilde C_1,\tilde C_2)$ for all $\tilde C_1,\tilde C_2$, Lemma \ref{weakconvaux} implies that $\lambda\geq\lambda(\eps,\delta)>0$ for large $n$. This implies the desired weak converse for $(\W,t_1,t_2,r)$ with codes$_\CONF$.
\end{proof}

\begin{proof}[Proof of Lemma \ref{diffsit}]
Let a code$_\CONF(n,M_1,M_2,C_1,C_2,\lambda)$ for $(\W,t_1,t_2,r)$ be given which has the form \eqref{codeseq} and which uses the conferencing function $g$. Without loss of generality, assume that 
\begin{equation}\label{assn}
  \frac{1}{n}\log \lvert T_1\rvert\lvert T_2\rvert\leq\delta.
\end{equation}
Set $\tilde\Gamma:=\Gamma\times T_1\times T_2$ and define
\[
  \pi_{\tau_1\tau_2}:\tilde\Gamma\rightarrow\Gamma
\]
to be the projection of $\tilde\Gamma$ onto $\Gamma\times\{(\tau_1,\tau_2)\}$. Further, define a conferencing MAC $\tilde g:[1,M_1]\times [1,M_2]\rightarrow\tilde\Gamma$ by
\[
  \tilde g(j,k)=(g(j,k,\tau_1,\tau_2))_{(\tau_1,\tau_2)\in T_1\times T_2}.
\]
As $g$ is the conferencing MAC of the code$_\CONF(n,M_1,M_2,C_1,C_2,\lambda)$, one obtains 
\begin{align*}
  \frac{1}{n}\log\lVert\tilde g_j\rVert&\leq C_2+\delta&\text{for all }j\in [1,M_1],\\
  \frac{1}{n}\log\lVert\tilde g_k\rVert&\leq C_1+\delta&\text{for all }k\in[1,M_2].
\end{align*}
This together with \eqref{assn} implies that $\tilde g$ is admissible for a code$_\AUX$ with conferencing capacities $C_\nu+\delta$. Further set $\tx{jk}{\tau_1\tau_2}:=\x{jk}{\tau_1\tau_2}$ and $\ty{jk}{\tau_1\tau_2}:=\y{jk}{\tau_1\tau_2}$ and $\tilde F_{jk}:=F_{jk}$. One checks immediately that the code thus defined is a code$_\AUX(n,R_1,R_2,C_1+\delta,C_2+\delta,\lambda)$ for $(\W,t,t,r)$. This proves the lemma.
\end{proof}

\subsubsection{The Weak Converse for the Auxiliary MAC}\label{subsubsect:convaux}

Here we prove Lemma \ref{weakconvaux}. Let $\delta>0$ be arbitrary and set
\[
  \tilde C_\nu:=C_\nu+\delta.
\]
Let a code$_\AUX(n,M_1,M_2,\tilde C_1,\tilde C_2,\lambda)$ be given which satisfies \eqref{l1aux}. We must show that there exists a $\lambda(\eps,\delta)$ such that $\lambda\geq\lambda(\eps,\delta)$ for large $n$.

Assume that the above code$_\AUX(n,M_1,M_2,\tilde C_1,\tilde C_2,\lambda)$ has the form
\[
 \{(\tx{jk}{\tau_1\tau_2},\ty{jk}{\tau_1\tau_2},\tilde F_{jk}^\rho:(\tau_1,\tau_2,\rho)\in T_1\times T_2\times R\},
\]
and uses the conferencing MAC $\tilde g$. We may assume that $\lambda\leq 1/4$, because otherwise, we are done. Consider a probability space $(\Omega,\mathcal{F},\PP)$ on which the following random variables are defined:
\begin{itemize}
  \item $(S_1,S_2)$ is uniformly distributed on $[1,M_1]\times[1,M_2]$,
  \item $G=\tilde g(S_1,S_2)$,
  \item for each $\tau=(\tau_1,\tau_2)\in T_1\times T_2$,
\[
  X^{\tau}=\tx{S_1S_2}{\tau},\quad Y^{\tau}=\ty{S_1S_2}{\tau},
\]
  \item for each $W\in\W_{\tau}$ a $Z^W$ taking values in $\Z^n$ such that for $\xx\in \X^n$, $\yy\in \Y^n$, $(j,k)\in[1,M_1]\times[1,M_2]$, and $\tilde\gamma\in\tilde\Gamma$,
\[
  \PP[Z^W=\zz\vert X^{\tau}=\xx,Y^{\tau}=\yy,S_1=j,S_2=k,G=\tilde\gamma]=W^n(\zz\vert\xx,\yy).
\]
\end{itemize}
Fix a $\tau\in T_1\times T_2$ and a $W\in\W_{\tau}$. By Fano's inequality,
\[
  H(S_1,S_2\vert Z^W)\leq\lambda\log(M_1M_2-1)+h(\lambda)=:\Delta_1,
\]
where $h$ denotes binary entropy. By the chain rule for entropy,
\begin{equation}\label{Fano}
  \Delta_1\geq H(S_1,S_2\vert Z^W)\geq H(S_1,S_2\vert Z^W,G)\geq H(S_1\vert S_2,Z^W,G)\vee H(S_2\vert S_1,Z^W,G).
\end{equation}
(Several rules for calculating with entropy are collected in \cite[Chapter 1.3]{CK}.) Using \eqref{Fano}, $M_1$ can be bounded via
\begin{equation}\label{boundR_1}
\begin{aligned}
  \log M_1&=H(S_1\vert S_2)\\
          &= I\rund{\left.S_1;Z^W,G\right\vert S_2}+H(S_1\vert S_2,Z^W,G)\\
          &\leq I\rund{\left.S_1;Z^W,G\right\vert S_2}+\Delta_1.
\end{aligned}
\end{equation}
One obtains an analogous bound on $M_2$,
\begin{equation}
  \log M_2\leq I\rund{\left.S_2;Z^W,G\right\vert S_1}+\Delta_1.\label{boundR_2}
\end{equation}
For $M_1M_2$ one has the bounds
\begin{equation}\label{boundsum}
\begin{aligned}
  \log M_1M_2&=H(S_1,S_2\vert G)\\
             &= I\rund{S_1,S_2;Z^W,G}+H(S_1,S_2\vert Z^W,G)\\
             &\leq I\rund{S_1,S_2;Z^W,G}+\Delta_1
\end{aligned}
\end{equation}
and
\begin{equation}\label{boundsum2}
\begin{aligned}
  \log M_1M_2&=H(S_1,S_2)\\
             &= I\rund{S_1,S_2;Z^W}+H(S_1,S_2\vert Z^W)\\
             &\leq I\rund{S_1,S_2;Z^W}+\Delta_1.
\end{aligned}
\end{equation}
Using the chain rule, one splits up the mutual information terms in the bounds \eqref{boundR_1}-\eqref{boundsum} into two terms each such that the channel only appears in the second one:
\begin{align*}
  I\rund{\left.S_1;Z^W,G\right\vert S_2}&=I\rund{\left.S_1;G\right\vert S_2}+I\rund{\left.S_1;Z^W\right\vert S_2,G},\\
  I\rund{\left.S_2;Z^W,G\right\vert S_1}&=I\rund{\left.S_2;G\right\vert S_1}+I\rund{\left.S_2;Z^W\right\vert S_1,G},\\
  I\rund{S_1,S_2;Z^W,G}&=I\rund{S_1,S_2;G}+I\rund{S_1,S_2;Z^W\vert G}.
\end{align*}
These mutual information terms and the one in \eqref{boundsum2} are bounded successively in the following. First, the terms not depending on the channel are considered. By the properties of $\tilde g$, if the value of $S_2$ is given, the random variable $G$ can assume at most $2^{n\tilde C_1}$ values, hence
\begin{align*}
  I\rund{\left.S_1;G\right\vert S_2}\leq\tilde C_1.
\end{align*}
An analogous argument shows
\[
  I\rund{\left.S_2;G\right\vert S_1}\leq\tilde C_2.
\]
Finally, as in Remark \ref{rem:detmac}, one sees that $H(G)\leq\tilde C_1+\tilde C_2$, so
\[
  I\rund{S_1,S_2;G}\leq\tilde C_1+\tilde C_2.
\]

Next we treat the remaining mutual information terms. Recall that for $(j,k)\neq(j',k')$, the corresponding codewords do not need to be distinct. This is a problem when $S_1,S_2$ are to be replaced by $X^\tau,Y^\tau$ in the expressions. Define $\Delta_2:=h(2\lambda)+2\lambda\log(M_1M_2)$. We next show that
\begin{align}
  I\rund{\left.S_1;Z^W\right\vert S_2,G}&\leq I\rund{\left.X^{\tau};Z^W\right\vert Y^{\tau},G}+\Delta_2,\label{I11}\\
  I\rund{\left.S_2;Z^W\right\vert S_1,G}&\leq I\rund{\left.Y^{\tau};Z^W\right\vert X^{\tau},G}+\Delta_2,\\
  I\rund{S_1,S_2;Z^W\vert G}&\leq I\rund{\left.X^{\tau},Y^{\tau};Z^W\right\vert G}+\Delta_2,\\
  I\rund{S_1,S_2;Z^W}&\leq I\rund{X^{\tau},Y^{\tau};Z^W}+\Delta_2.\label{I22}
\end{align}
This allows us to do the replacement and to control the error incurred by the replacement. In order to show \eqref{I11}-\eqref{I22}, we write
\begin{align*}
  I(Z^W;S_1\vert S_2,G)&=H(Z^W\vert S_2,G)-H(Z^W\vert S_1,S_2,G),\\
  I(Z^W;S_2\vert S_1,G)&=H(Z^W\vert S_1,G)-H(Z^W\vert S_1,S_2,G),\\
  I(Z^W;S_1,S_2\vert G)&=H(Z^W\vert G)-H(Z^W\vert S_1,S_2,G),\\
  I(Z^W;S_1,S_2)       &=H(Z^W)-H(Z^W\vert S_1,S_2).
\end{align*}
One has $H(Z^W\vert S_2,G)\leq H(Z^W\vert Y^\tau,G)$ and $H(Z^W\vert S_1,G)\leq H(Z^W\vert X^\tau,G)$, as $(X^\tau,G)$ is a function of $(S_1,G)$ and $(Y^\tau,G)$ is a function of $(S_2,G)$. Thus in order to show \eqref{I11}-\eqref{I22}, we need to bound the distance of $H(Z^W\vert S_1,S_2,G)$ from $H(Z^W\vert X^\tau,Y^\tau,G)$ and of $H(Z^W\vert S_1,S_2)$ from $H(Z^W\vert X^\tau,Y^\tau)$. 

\begin{lem}
  One has 
\begin{align*}
  H(Z^W\vert S_1,S_2,G)&\geq H(Z^W\vert X^\tau,Y^\tau,G)-\Delta_2,\\
  H(Z^W\vert S_1,S_2)&\geq H(Z^W\vert X^\tau,Y^\tau)-\Delta_2.
\end{align*}
\end{lem}

\begin{proof}
Note that as $G$ is a function of $(S_1,S_2)$, 
\[
  H(Z^W\vert S_1,S_2,G)=H(Z^W,S_1,S_2,G)-H(S_1,S_2)-H(G\vert S_1,S_2)=H(Z^W,S_1,S_2,G)-H(S_1,S_2)
\]
and 
\[
  H(Z^W\vert S_1,S_2)=H(Z^W,S_1,S_2)-H(S_1,S_2).
\]
Now $(Z^W,X^\tau,Y^\tau)$ is a function of $(Z^W,S_1,S_2)$ and $(Z^W,X^\tau,Y^\tau,G)$ is a function of $(Z^W,S_1,S_2,G)$, so one has
\begin{align*}
  H(Z^W,S_1,S_2,G)&\geq H(Z^W,X^\tau,Y^\tau,G),\\
  H(Z^W,S_1,S_2)&\geq H(Z^W,X^\tau,Y^\tau).
\end{align*}
Hence it suffices to show 
\begin{equation}\label{telos}
  H(S_1,S_2)\leq H(X^\tau,Y^\tau)+\Delta_2.
\end{equation}

  Set 
\[
  \mathcal{G}_W:=\{(j,k):W^n((F_{jk}^W)^c\vert\x{jk}\tau,\y{jk}\tau)<1/2\}
\]
and set $\mathcal{B}_W:=([1,M_1]\times[1,M_2])\setminus\mathcal{G}_W$. From 
\[
  \lambda\geq\frac{1}{M_1M_2}\sum_{j,k}W^n((F_{jk}^W)^c\vert\x{jk}\tau,\y{jk}\tau)\geq\frac{\lvert\mathcal{B}_W\rvert}{2M_1M_2}
\]
it follows that $\lvert\mathcal{B}_W\rvert\leq2\lambda M_1M_2$. Now if $(j,k),(j',k')\in\mathcal{G}_W$, then $(\x{jk}\tau,\y{jk}\tau)\neq(\x{j'k'}\tau,\y{j'k'}\tau)$, because otherwise one would obtain a contradiction to the disjointness of $F_{jk}^W$ and $F_{j'k'}^W$. We introduce the random variable $Q=1_{\mathcal{G}_W}(S_1,S_2)$ which equals 1 if $(S_1,S_2)\in\mathcal{G}_W$ and 0 else. The above bound on the size of $\mathcal{B}_W$ implies $H(Q)\leq h(2\lambda)$ for $\lambda<1/2$. Therefore
\begin{align*}
  H(S_1,S_2)&=H(S_1,S_2,Q)\\
            &\leq H(S_1,S_2,Q)-H(Q)+h(2\lambda)\\
            &\leq H(S_1,S_2\vert Q)+h(2\lambda).
\end{align*}
The assignment of message pairs to codewords is unique on $\mathcal{G}_W$, so
\begin{align*}
  H(S_1,S_2\vert Q)&=H(X^\tau,Y^\tau\vert Q=1)\PP[Q=1]+H(S_1,S_2\vert Q=0)\PP[Q=0]\\
                   &\leq H(X^\tau,Y^\tau)+2\lambda\log (M_1M_2).
\end{align*}
Altogether this shows \eqref{telos}, and thus the lemma.
\end{proof}


Thus \eqref{I11}-\eqref{I22} is established. The next goal is to obtain a single-letter representation of the right-hand terms in \eqref{I11}-\eqref{I22}. This is done by several applications of the chain rules. Set
\begin{align*}
  X^{\tau}&=(X^{\tau}_{1},\ldots,X^{\tau}_{n}),\\
  Y^{\tau}&=(Y^{\tau}_{1},\ldots,Y^{\tau}_{n}),\\
  Z^W&=(Z^{W}_{1},\ldots,Z^{W}_{n}).
\end{align*}
Further, set
\[
  Z^{W}_{[1,m]}:=(Z^{W}_{1},\ldots,Z^{W}_{m})\quad\text{for }m=1,\ldots,n.
\]
One has
\begin{align*}
  I(X^{\tau};Z^W\vert Y^{\tau},G)&=\sum_{m=1}^n\menge{H(Z^{W}_{m}\vert Y^{\tau},G,Z^{W}_{[1,m-1]})-H(Z^{W}_{m}\vert X^{\tau},Y^{\tau},G,Z^{W}_{[1,m-1]})}.
\end{align*}
$(Y^{\tau}_{m},G)$ is a function of $(Y^{\tau},G,Z^{W}_{[1,m-1]})$, so
\[
  H(Z^{W}_{m}\vert Y^{\tau},G,Z^{W}_{[1,m-1]})\leq H(Z^{W}_{m}\vert Y^{\tau}_{m},G).
\]
Further as the channel is memoryless,
\begin{align*}
  &H(Z^{W}_{m}\vert X^{\tau},Y^{\tau},G,Z^{W}_{[1,m-1]})\\
  &=-I\rund{Z^{W}_{m};Z^{W}_{[1,m-1]}\left\vert X^{\tau},Y^{\tau},G\right.}+H(Z^{W}_m\vert X^{\tau},Y^{\tau},G)\\
  &=H(Z^{W}_{m}\vert X^{\tau}_{m},Y^{\tau}_{m},G).
\end{align*}
Hence
\begin{align*}
  I(X^{\tau};Z^W\vert Y^{\tau},G)&\leq\sum_{m=1}^n\menge{H(Z^{W}_{m}\vert Y^{\tau}_{m},G)-H(Z^{W}_{m}\vert X^{\tau}_{m},Y^{\tau}_{m},G)}\\
  &=\sum_{m=1}^nI(Z^{W}_{m};X^{\tau}_{m}\vert Y^{\tau}_{m},G).
\end{align*}
In an analogous manner, one shows that
\begin{align*}
  I(Y^{\tau};Z^W\vert X^{\tau},G)
  \leq\sum_{m=1}^nI(Z^{W}_{m};Y^{\tau}_{m}\vert X^{\tau}_{m},G).
\end{align*}
Further, with the same arguments as above,
\begin{align*}
   I(Z^W;X^{\tau},Y^{\tau}\vert G) &=\sum_{m=1}^n\menge{H(Z^{W}_{m}\vert G,Z^{W}_{[1,m-1]})-H(Z^{W}_{m}\vert X^{\tau},Y^{\tau},G,Z^{W}_{[1,m-1]})}\\
                                     &\leq\sum_{m=1}^n\menge{H(Z^{W}_{m}\vert G)-H(Z^{W}_{m}\vert X^{\tau}_{m},Y^{\tau}_{m},G)}\\
                                     &=\sum_{m=1}^nI(Z^{W}_{m};X^{\tau}_{m},Y^{\tau}_{m},G).
\end{align*}
Finally,
\[
  I(Z^W;X^{\tau},Y^{\tau})\leq\sum_{m=1}^nI(Z^{W}_{m};X^{\tau}_{m},Y^{\tau}_{m}).
\]

Now we define the random variables that will be used for the single-letter characterization. Let $U$ take values in $[1,n]\times\tilde\Gamma$, $X_\tau$ in $\X$, $Y_\tau$ in $\Y$, and $Z_W$ in $\Z$, with
\begin{align*}
  \PP[U=(m,\tilde\gamma)]&=\frac{1}{n}\frac{\betr{\menge{(j,k):\tilde g(j,k)=\tilde\gamma}}}{M_1M_2}=:p_{0}(m,\tilde\gamma);\\
  \PP[X_\tau=x\vert U=(m,\tilde\gamma)]&=\frac{\betr{\menge{(j,k):\tx{jk,m}{\tau}=x}}}{\betr{\menge{(j,k):\tilde g(j,k)=\tilde\gamma}}}=:p_{1\tau}(x\vert(m,\tilde\gamma));\\
  \PP[Y_\tau=y\vert U=(m,\tilde\gamma)]&=\frac{\betr{\menge{(j,k):\ty{jk,m}{\tau}=y}}}{\betr{\menge{(j,k):\tilde g(j,k)=\tilde\gamma}}}=:p_{2\tau}(y\vert(m,\tilde\gamma));
\end{align*}
and 
\[
  \PP[Z_W=z\vert U=(m,\tilde\gamma),X_\tau=x,Y_\tau=y]=W_W(z\vert x,y).
\]
Note that $\U:=\text{support}(p_0)\subset[1,n]\times\tilde\Gamma$ is a finite set, that $p_0\in \P(\U)$, that $p_{1\tau}\in\mathcal{K}(\X\vert \U)$ and that $p_{2\tau}\in\mathcal{K}(\Y\vert \U)$. Further,
\begin{align*}
  p_0(m,\tilde\gamma)&=\frac{1}{n}\PP[G=\tilde\gamma],\\
  p_{1\tau}(x\vert(m,\tilde\gamma))&=\PP[X^{\tau}_{m}=x\vert G=\tilde\gamma],\\
  p_{2\tau}(y\vert(m,\tilde\gamma))&=\PP[Y^{\tau}_{m}=y\vert G=\tilde\gamma].
\end{align*}
Combining the above equalities and inequalities, this implies that
\begin{align*}
  \frac{1}{n}I(S_1;Z^W\vert S_2,G)\leq\frac{1}{n}\sum_{m=1}^nI\rund{\left.Z^{W}_{m};X^{\tau}_{m}\right\vert Y^{\tau}_{m},G}&=I(Z_W;X_\tau\vert Y_\tau,U);\\
  \frac{1}{n}I(S_2;Z^W\vert S_1,G)\leq\frac{1}{n}\sum_{m=1}^nI\rund{\left.Z^{W}_{m};Y^{\tau}_{m}\right\vert X^{\tau}_{m},G}&=I(Z_{W};Y_\tau\vert X_\tau,U);\\
  \frac{1}{n}I(Z^W;S_1,S_2\vert G)\leq\frac{1}{n}\sum_{m=1}^nI\rund{\left.Z^{W}_{m};X^{\tau}_{m},Y^{\tau}_{m}\right\vert G}&=I(Z_W;X_\tau,Y_\tau\vert U);\\
  \frac{1}{n}I(Z^W;S_1,S_2)\leq\frac{1}{n}\sum_{m=1}^nI\rund{Z^{W}_{m};X^{\tau}_{m},Y^{\tau}_{m}}&=I(Z_W;X_\tau,Y_\tau).
\end{align*}

Thus for every $\tau\in T_1\times T_2$ and every $W\in \W_\tau$, using \eqref{Fano}-\eqref{I22} and recalling the definitions of $\Delta_1$ and $\Delta_2$, one has the bounds
\begin{align}
  \frac{1}{n}\log M_1&\leq \tilde C_1+I(Z_W;X_\tau\vert Y_\tau,U)+\frac{1}{n}\Delta;\label{primus}\\
  \frac{1}{n}\log M_2&\leq \tilde C_2+I(Z_W;Y_\tau\vert X_\tau,U)+\frac{1}{n}\Delta;\\
  \frac{1}{n}\log M_1M_2&\leq\bigl\{\bigl(\tilde C_1+\tilde C_2+I(Z_W;X_\tau,Y_\tau\vert U)\bigr)\wedge I(Z_W;X_\tau,Y_\tau)\bigr\}+\frac{1}{n}\Delta.\label{tertius}
\end{align}
On the other hand, the validity of \eqref{l1aux} implies that there is a $\tau\in T_1\times T_2$ and a $W\in \W_\tau$ such that one of the following inequalities holds:
\begin{align}
  \frac{1}{n}\log M_1&\geq \tilde C_1+I(Z_W;X_\tau\vert Y_\tau,U)+\eps;\label{unus}\\
  \frac{1}{n}\log M_2&\geq \tilde C_2+I(Z_W;Y_\tau\vert X_\tau,U)+\eps;\label{duus}\\
  \frac{1}{n}\log M_1M_2&\geq\bigl\{\bigl(\tilde C_1+\tilde C_2+I(Z_W;X_\tau,Y_\tau\vert U)\bigr)\wedge I(Z_W;X_\tau,Y_\tau)\bigr\}+\eps.\label{trius}
\end{align}
According to which of \eqref{unus}-\eqref{trius} holds, we distinguish between three cases. In order to simplify notation, we write
\[
  \bigl(\tilde C_1+\tilde C_2+I(Z_W;X_\tau,Y_\tau\vert U)\bigr)\wedge I(Z_W;X_\tau,Y_\tau)=:I_0.
\]

\textit{Case 1:} \eqref{trius} holds. Then comparing \eqref{trius} with \eqref{tertius} yields
\[
  1-2\lambda\leq\frac{I_0+\frac{2}{n}\log 2}{I_0+\eps}.
\]
But if $\lambda$ is chosen small enough, this gives a contradiction if $n$ is large depending on $\eps$ and $\lambda$. Thus for small $\lambda=\lambda(\eps)$ and large $n=n(\lambda,\eps)$, there can be no code$_\CONF(n,M_1,M_2,\tilde C_1,\tilde C_2),\lambda)$ satisfying \eqref{trius}.

\textit{Case 2:} \eqref{trius} does \textit{not} hold, but \eqref{unus} holds. Together with \eqref{primus}, the fact that \eqref{trius} does not hold implies
\[
  \frac{1}{n}\log M_1\leq\tilde C_1+I(Z_W;X_\tau\vert Y_\tau,U)+\frac{2\log 2}{n}+2\lambda(I_0+\eps).
\]
Then using \eqref{primus}, we obtain
\[
  \lambda\geq\frac{\eps}{2(I_0+\eps)}+\frac{\log 2}{n(I_0+\eps)}.
\]
Thus for large $n$, there can be no code$_\AUX(n,M_1,M_2,\tilde C_1,\tilde C_2,\lambda)$ satisfying \eqref{unus} if $\lambda$ is too small.

\textit{Case 3:} \eqref{trius} does \textit{not} hold, but \eqref{duus} holds. Analogous to case 2.

 And this proves the weak converse for the auxiliary MAC.

\subsection{The Converse for the MAC with Common Message}

We restrict ourselves here to describing the setting that is the starting point for the weak converse and apply Fano's inequality. The rest is single-letterization of mutual information terms and similar to what was done in \ref{subsubsect:convaux}. We assume full CSIR again. 

For a $\lambda>0$, let a code$_\CM$ $(n,M_0,M_1,M_2,\lambda)$ be given with the form \eqref{codeseq} and conferencing MAC $g$. Let a probability space $(\Omega,\mathcal{F},\PP)$ be given on which the following random variables are defined:
\begin{itemize}
  \item $S_0$ uniformly distributed on $[1,M_0]$,
  \item $S_1$ uniformly distributed on $[1,M_1]$ given $S_0$ and $S_2$ uniformly distributed on $[1,M_2]$ given $S_0$,
  \item for every $(\tau_1,\tau_2)\in T_1\times T_2$, 
\[
  X^{\tau_1}=\x{ij}{\tau_1},\quad Y^{\tau_2}=\y{ik}{\tau_2},
\]
  \item for every $W\in \W_\tau$ a random variable $Z^W$ such that
\[
  \PP\eckig{Z^W=\zz\vert X^{\tau_1}=\xx,Y^{\tau_2}=\yy,S_0=i,S_1=j,S_2=k}=W^n(\zz\vert\xx,\yy)
\]
for all $\xx\in\X^n,\yy\in\Y^n$.
\end{itemize}
If $W\in \W_\tau$, the definition of the code$_\CM$ and Fano's inequality imply
\begin{align*}
  \lambda\log(M_0M_1M_2-1)+h(\lambda)&\geq H(S_0,S_1,S_2\vert Z^W)\\
                                     &\geq H(S_1,S_2\vert Z^W,S_0)\\
                                     &\geq H(S_1\vert Z^W,S_0,S_2)\vee H(S_2\vert Z^W,S_0,S_1).
\end{align*}
From this point, replacing the message variables by the codeword variables and the single-letterization are very similar to the one done in the converse for the MAC with Conferencing Encoders, so we omit them. Thus the weak converse for Theorem \ref{thmcm} is proved.

\section{Application and Numerical Example}\label{sect:numerik}

\subsection{Applications in Wireless Networks}\label{subsect:numapp}

It was noted in the Introduction that the information-theoretic compound MAC with conferencing encoders can be used to analyze ``virtual MISO systems''. We now give the informal description of a simplified wireless ``virtual MISO'' network which we will then translate into our setting of compound MAC with conferencing encoders. Assume that one data stream intended for one receiving mobile terminal is to be transmitted. Two base stations, which are placed at spatially remote positions, are used to send the data to the destination. Assume that the base stations are fed by a central network node with their part of the information which is to be transmitted. At the receiver, the two streams received from the two base stations are then combined to form the original data stream. The question arises how the original data stream should be distributed by the central node in order to achieve a good performance. We will assume that the central node has the combined CSIT of both transmitters, which could for example be achieved by feedback. The network is pictured in Figure \ref{real_netw}.

\begin{figure}
  \begin{center}
    \includegraphics{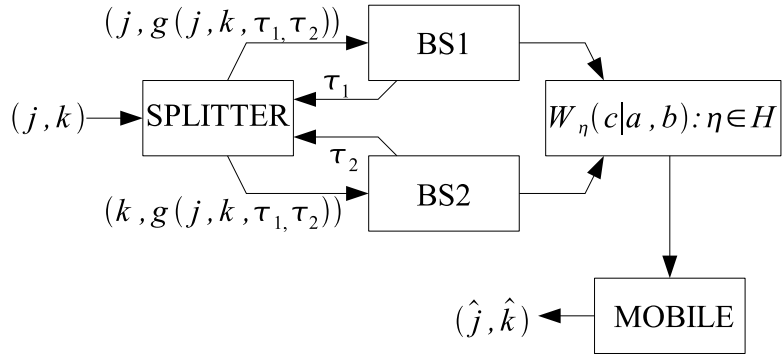}
  \end{center}
  \caption{A central node distributing one data stream to two senders.}
  \label{real_netw}
\end{figure}

The answer to this problem can be given immediately once one has translated the question into the setting of compound MAC with generalized conferencing. If the data stream is not split at all, but both senders know the complete message and also have the other transmitter's CSIT, then the full-cooperation sum capacity is achieved, i.e. the capacity of the system where the senders and the central node are all at the same location. The drawback of this scheme is that the capacity of each of the links from the central node to the base stations must be at least the full-cooperation sum capacity. The other extreme is if the central node just splits each message from the data stream into two components. Then the overhead which needs to be transmitted to the corresponding sender in addition to its message component is minimized. However, the full-cooperation sum capacity will not be achieved in general. The goal should be to find the minimal amount of overhead which suffices to achieve a good performance.

From Theorem \ref{thmconf} it follows that it suffices for the splitter to send to the first base station, in addition to the first component of the message, the one-shot Willems conferencing function value attained by the message of the second component and the second sender's CSIT. The analogous statement holds for the overhead for the second sender. The sum of the overhead rates required to achieve the full-cooperation sum capacity can be seen from Corollary \ref{sumcapcor}. See also the following numerical example.

\subsection{Numerics}

We present a simple example of a rate region for the MAC with conferencing encoders. Assume $\X=\Y=\Z=\{0,1\}$. Let $\W$ consist of the stochastic matrices
\[
  W_1=\begin{pmatrix}
        .9&.1\\
        .4&.6\\
        .6&.4\\
         0& 1
      \end{pmatrix}\quad\text{and}\quad
  W_2=\begin{pmatrix}
        .9&.1\\
        .6&.4\\
        .4&.6\\
         0& 1
      \end{pmatrix}.
\]
Here, the output distribution corresponding to the input combination $(x,y)$ is written in row $2x+y+1$. 

In Figure \ref{fig:ratenregion}, different capacity regions are pictured. $W_1$ and $W_2$ denote the capacity regions of the MACs given by $W_1$ and $W_2$, respectively, without cooperation. Their intersection is the capacity region of the compound channel consisting of $W_1$ and $W_2$, where the exact channel is known at the transmitter.  The capacity region in the case of no CSIT is  shown for no cooperation ($C_{11}=0,C_{12}=0$). Note that absence of cooperation makes the region strictly smaller. $C_{21}$ and $C_{22}$ have been chosen such that their sum is the minimal $C_1+C_2$ achieving the optimal sum capacity:
\[
  C_{2\nu}=\frac{1}{2}\bigl(C^\infty-\max_\mathcal{M}\;\min_{i=1,2}I(Z_i;X,Y\vert U)\bigr)\approx .29
\]
$C_{31}=.33$ has been chosen as .1 minus the minimal $C_1$ such that the first user achieves the maximal possible rate, and $C_{32}=.43$ has been chosen as the minimal $C_2$ such that the second user achieves the maximal possible rate. Finally, ``full coop.'' denotes the rate region which can be achieved by full cooperation. As noted in Corollary \ref{sumcapcor}, it can already be achieved with $C_1=.47$ and $C_2=.47$.

\begin{figure}
  \begin{center}
    \includegraphics[width=.7\linewidth]{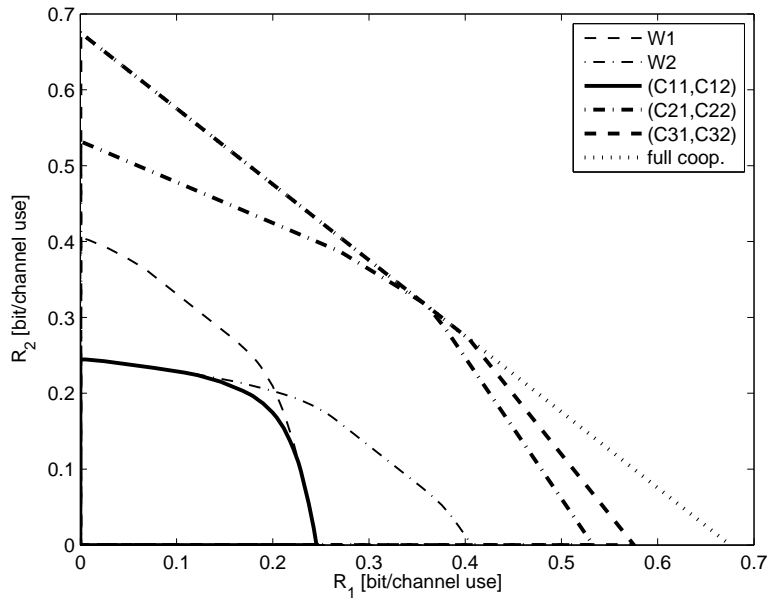}
    \caption{The capacity regions for the conferencing capacity pairs $(C_{11},C_{12}) = (0,0)$, $(C_{21},C_{22})=(.29,.29)$, and $(C_{31},C_{32})=(.33,.43)$.}
    \label{fig:ratenregion}
  \end{center}
\end{figure}

\section{Conclusion and Outlook}

We have derived the capacity regions of two information-theoretic compound multiple access channels: the compound multiple-access channel with common message and the compound multiple-access channel with conferencing encoders, where conferencing can be done about messages and channel state information. The channel with common message, aside from the interest it has on its own, was used to derive the capacity region of the channel with conferencing encoders. The latter channel can be applied in the rigorous information-theoretic analysis of certain wireless cellular networks which use base station cooperation in order to transmit data to one mobile receiver. One can derive the exact amount of base station cooperation that is needed in order to achieve the sum capacity and the capacity region as would be achievable if the base stations were at the same location and could thus be regarded as forming a ``virtual MISO system''.

This analysis was motivated by recent developments in the design of cellular systems. As interference is the main limiting factor in the performance of such systems, research has recently focused on methods of controlling interference in order to meet the requirements for future wireless systems such as LTE-Advanced. Much of the literature which has contributed to this research uses strict assumptions that will not generally be met in reality. Assuming limited base station cooperation and channel uncertainty in this paper, we tried to obtain a more appropriate description of real situations.

Note that we did not address the issue of unknown out-of-network interference. This is a problem for real networks. Different systems operating in the same frequency band and operated by different providers who do not jointly design their systems interfere each other. This happens, e.g., when Wireless Local Area Network (WLAN)-systems are located close to each other. Future work will be to model this information-theoretically. The appropriate model is to take multiple-access channels with conferencing encoders. However in this case, channel uncertainty should not be included by considering a compound channel, but rather, the model best describing reality is the arbitrarily varying channel. In such a channel, the transmission probabilities can change for each channel use in a way unknown to the encoder. (This is just the way unknown interference acts on channels.) Ahlswede's robustification technique \cite{A} shows how to construct codes for arbitrarily varying channels from codes for compound channels. Hence from that point of view, the work done in the present paper can also be regarded as a preliminary needed for the analysis of arbitrarily varying multiple-access channels with conferencing encoders.

\appendix

Here, we include some technical lemmas concerning typical sequences.

\begin{lem}\label{Igor}
  a) Let $\X$ be a finite set. Let $p,\tilde p\in\P(\X)$. Let $0<\delta<1/(2\lvert \X\rvert)$. Then, for all $n\in\mathbb{N}$, for every $\xx\in T_{\tilde p,\delta}^n$,
\[
  p^n(\xx)\leq 2^{-n(H(\tilde p)-\phi_1(\lvert \X\rvert,\delta))}.
\]
$\phi_1$ is a universal function (i.~e. independent of everything), positive if $\lvert \X\rvert\geq 1$ and $0<\delta<1$, and for all values of $\lvert \X\rvert$, one has $\lim_{\delta\rightarrow 0}\phi_1(\lvert \X\rvert,\delta)=0$.

  b) Let $\X,\Y$ be finite sets. Let $p\in\P(\X)$ and $W,\tilde W$ stochastic matrices with input alphabet $\X$ and output alphabet $\Y$. Let $0<\delta<1/(2\lvert \X\rvert\lvert \Y\rvert)$. Let $\tilde r\in\P(\X\times \Y)$ be the joint distribution corresponding to $p$ and $\tilde W$. 
Then, for all $n\in\mathbb{N}$, for all $(\xx,\yy)\in T_{\tilde r,\delta}^n$,
\[
  W^n(\yy\vert\xx)\leq 2^{-n(H(\tilde W\vert p)-\phi_2(\delta,\lvert \X\rvert,\lvert \Y\rvert))}.
\]
$\phi_2$ is a universal function (i.~e. independent of everything), positive if $\lvert \X\rvert,\lvert \Y\rvert\geq 1$, $0<\delta<1$, and for arbitrary $\lvert \X\rvert,\lvert \Y\rvert$, one has $\lim_{\delta\rightarrow 0}\phi_2(\delta,\lvert \X\rvert,\lvert \Y\rvert)=0$ .
\end{lem}

\begin{proof}
  This is essentially \cite[Lemma 1.2.6 and 1.2.7]{CK}.
\end{proof}

\begin{lem}\label{Pinsker}
  Let $\X$ be a finite set and let $p\in\P(\X)$. Then, there is a universal constant $c>0$ such that 
\[
  p^n((T_{p,\delta}^n)^c)\leq (n+1)^{\lvert \X\rvert}2^{-nc\delta^2}.
\]
\end{lem}

\begin{proof}
  This is exactly \cite[Lemma III.1.3]{S}
\end{proof}

The next lemma is not used in the text. However, it is used in the proof of Lemma \ref{unionlemma}, which we will prove. For $\xx\in\X^n$ and $W\in\mathcal{W}(\Y\vert\X)$, denote by $T_{W,\delta}^n(\xx)$ the set of $\yy\in\Y^n$ that are $W$-generated by $\xx$ with constant $\delta$ (cf. \cite[Definition 1.2.9]{CK}). 

\begin{lem}\label{lem:typecard}
  Let $\X,\Y$ be finite sets. Let $p\in\P(\X)$ and $W\in\mathcal{W}(\Y\vert \X)$. Let $0<\delta<1/(2\lvert \X\rvert)$. Then for any $\xx\in T_{p,\delta}^n$,
\[
  \lvert T_{W,\delta}^n(\xx)\rvert\leq (n+1)^{\lvert\X\rvert}2^{n(H(W\vert p)+\phi_3(\lvert \X\rvert,\lvert \Y\rvert,\delta))}.
\]
$\phi_3$ is a universal function (i.~e. independent of everything), positive if $\lvert \X\rvert,\lvert \Y\rvert\geq 1$, $0<\delta<1$, and for arbitrary $\lvert \X\rvert,\lvert \Y\rvert$, one has $\lim_{\delta\rightarrow 0}\phi_3(\delta,\lvert \X\rvert,\lvert \Y\rvert)=0$.
\end{lem}

\begin{proof}
  This is essentially \cite[Lemma 1.2.13]{CK}.
\end{proof}

The following lemma was already used in \cite{J}. A slightly different form was proved in \cite{A}. As it is non-standard, we give a proof here. 

\begin{lem}\label{unionlemma}
  Let $\W$ be a nonempty set, and let $\X$ and $\Y$ be finite sets. For every $W\in \W$, let $W\in\mathcal{W}(\Y\vert\X)$. Let $\xx\in \X^n$ and $p\in\P(\X)$. Define the probability measure $q_W$ on $\Y\times \X$ by 
\[
  q_W(y,x)=W(y\vert x)p(x).
\]
Then,
\begin{equation}
  \left\lvert\bigcup_{W\in \W}T_{q_W,\delta}^n\rvert_\xx\right\rvert\leq2^{n(\sup_{W\in \W}H(W\vert p)+\psi(\delta))}
\end{equation}
for some universal positive function $\psi$ which tends to 0 as $\delta\rightarrow 0$.
\end{lem}
\begin{proof}
Denote by $T(\W)$ all the joint types $\hat q$ in $\Y\times \X$ such that there is an $W\in \W$ with
\[
  \lvert q_W(y,x)-\hat q(y,x)\rvert<\delta.
\]
Every $T_{q_W,\delta}^n$ can be written as the union of some $T_{\hat q}^n$, where $\hat q\in T(\W)$. Hence
\begin{equation}\label{erste}
  \left\lvert\bigcup_{W\in \W}T_{q_W,\delta}^n\rvert_\xx\right\rvert\leq\left\lvert\bigcup_{\hat q\in T(\W)}T_{\hat q}^n\rvert_\xx\right\rvert.
\end{equation}
As there are at most $(n+1)^{\lvert \X\rvert\lvert \Y\rvert}$ different joint types in $\Y^n\times \X^n$, this is smaller than
\begin{equation}\label{zweite}
  (n+1)^{\lvert \X\rvert\lvert \Y\rvert}\max_{\hat q\in T(\W)}\left\lvert T_{\hat q}^n\rvert_\xx\right\rvert.
\end{equation}
Without loss of generality, we can assume that the union on the left side of \eqref{erste} is nonempty. Hence there is an $W\in \W$ with $T_{q_W,\delta}^n\vert_\xx\neq\varnothing$, so $\xx\in T_{p,\lvert \Y\rvert\delta}^n$. This implies for any $q_W$ which is close to $\hat q$ (in the sense of the definition of $T(\W)$) and for all $x\in \X$ and $y\in \Y$
\begin{align*}
  &\mathrel{\hphantom{\leq}}\lvert\hat q(y,x)-W(y\vert x)p_\xx(x)\vert\\
  &\leq\lvert\hat q(y,x)-q_W(y,x)\rvert+\lvert q_W(y,x)-W(y\vert x)p_\xx(x)\rvert\\
  &\leq \delta+W(y\vert x)\lvert p(x)-p_\xx(x)\rvert\\
  &\leq (\lvert \Y\rvert +1)\delta.
\end{align*}
Thus $T_{\hat q}^n\vert_\xx\subset T_{W,(\lvert \Y\rvert+1)\delta}^n(\xx)$. By Lemma \ref{lem:typecard}, we conclude, using \eqref{erste} and \eqref{zweite}, that 
\[
  \left\lvert\bigcup_{W\in \W}T_{q_W,\delta}^n\rvert_\xx\right\rvert\leq(n+1)^{\lvert \X\rvert\lvert \Y\rvert}\sup_{W\in \W}2^{n(H(W\vert p)-\phi(\delta))},
\]
which finishes the proof.
\end{proof}

\begin{biography}{Moritz Wiese}
  (S'09) received the Dipl.-Math. degree in mathematics from the university of Bonn, Germany, in 2007. He has been pursuing the PhD degree since then. From 2007 to 2010, he was a research assistant at the Heinrich-Hertz-Lehrstuhl f\"ur Mobilkommunikation, Technische Universit\"at Berlin, Germany. Since 2010, he is a research and teaching assistant at the Lehrstuhl f\"ur Theoretische Informationstechnik, Technische Universit\"at M\"unchen, Munich, Germany.
\end{biography}

\begin{biography}{Holger Boche}
  (M’04-SM’07-F'11) received the Dipl.-Ing. and Dr.-Ing. degrees in electrical engineering from the Technische Universitaet Dresden, Dresden, Germany, in 1990 and 1994, respectively. He graduated in mathematics from the Technische Universitaet Dresden in 1992. From 1994 to 1997, he did postgraduate studies in mathematics at the Friedrich-Schiller Universität Jena, Jena, Germany. He received his Dr.Rer.Nat. degree in pure mathematics from the Technische Universitaet Berlin, Berlin, Germany, in 1998. In 1997, he joined the Heinrich-Hertz-Institut (HHI) für Nachrichtentechnik Berlin, Berlin, Germany. Since 2002, he has been a Full Professor for mobile communication networks with the Institute for Communications Systems, Technische Universität Berlin. In 2003, he became Director of the Fraunhofer German-Sino Lab for Mobile Communications, Berlin, Germany, and since 2004 he has also been Director of the Fraunhofer Institute for Telecommunications (HHI), Berlin, Germany. Since, October 2010 he is with the Institute of Theoretical Information Technology and Full Professor at the Technical University of Munich, Munich, Germany. He was a Visiting Professor with the ETH Zurich, Zurich, Switzerland, during the 2004 and 2006 Winter terms, and with KTH Stockholm, Stockholm, Sweden, during the 2005 Summer term. Prof. Boche is a Member of IEEE Signal Processing Society SPCOM and SPTM Technical Committee. He was elected a Member of the German Academy of Sciences (Leopoldina) in 2008 and of the Berlin Brandenburg Academy of Sciences and Humanities in 2009. He received the Research Award “Technische Kommunikation” from the Alcatel SEL Foundation in October 2003, the “Innovation Award” from the Vodafone Foundation in June 2006, and the Gottfried Wilhelm Leibniz Prize from the Deutsche Forschungsgemeinschaft (German Research Foundation) in 2008. He was co-recipient of the 2006 IEEE Signal Processing Society Best Paper Award and recipient of the 2007 IEEE Signal Processing Society Best Paper Award. 
\end{biography}

\begin{biography}{Igor Bjelakovi\'c}
  received the Dipl. Phys. degree in physics and Dr.rer.nat. degree in mathematics from the Technische Universit\"at Berlin, Germany, in 2001 and 2004, respectively. He was a Postdoctoral Researcher at the Heinrich-Hertz-Chair for Mobile Communications and the Department of Mathematics at the Technische Universit\"at Berlin. He is now with Technische Universit\"at M\"unchen, Lehrstuhl f\"ur Theoretische Informationstechnik.
\end{biography}

\begin{IEEEbiography}{Volker Jungnickel}
received a Dipl.-Phys. and Dr.~rer.~nat. (Ph.D.) degree in physics from Humboldt University in Berlin, Germany, in 1992 and 1995, respectively. He joined Fraunhofer Heinrich Hertz Institute (HHI) in Berlin, Germany, in 1997. He has contributed to high-speed indoor wireless infrared links, 1 Gbit/s MIMO-OFDM radio transmission and initial field trials for LTE and LTE-Advanced. Volker is a lecturer for wireless communications at University of Technology in Berlin and head of the cellular radio research team at HHI.
\end{IEEEbiography} 

\newpage

\subsection*{List of Figures:}

\begin{itemize}
  \item Figure 1: The MAC with Common Message
  \item Figure 2: The MAC with Conferencing Encoders
  \item Figure 3: A central node distributing one data stream to two senders.
  \item Figure 4: The capacity regions for the conferencing capacity pairs $(C_{11},C_{12})=(0,0)$, $(C_{21},C_{22})=(.29,.29)$, and $(C_{31},C_{32})=(.33,.43)$.
\end{itemize}

\end{document}